\documentclass[journal]{IEEEtran}

\usepackage{amsmath,amssymb,amsfonts}
\usepackage{graphicx}
\usepackage{cite}
\usepackage{booktabs}
\usepackage{array}
\usepackage{tabularx}
\usepackage[caption=false,font=footnotesize]{subfig}
\usepackage{stfloats}
\usepackage{url}
\usepackage{float}
\usepackage{bm}
\usepackage{multirow}
\usepackage[hidelinks]{hyperref}

\newtheorem{theorem}{Theorem}
\newtheorem{lemma}{Lemma}
\newtheorem{corollary}{Corollary}

\title{Temperature-Dependent Multiphysics Modeling of Additive Friction Stir Deposition Using Multi-Task Coupled Physics-Informed Neural Networks}

\author{Dhrubajyoti~Gupta,
Nikhil~Gotawala,
Raghav~Gnanasambandam,
Rohit~Kannan,
Jian~Yu,
Hang~Z.~Yu,
Zhenyu~James~Kong%
\thanks{D. Gupta, R. Kannan, and Z. J. Kong are with the Grado Department of Industrial and Systems Engineering, Virginia Tech, Blacksburg, VA 24061 USA.}%
\thanks{N. Gotawala and H. Z. Yu are with the Department of Materials Science and Engineering, Virginia Tech, Blacksburg, VA 24061 USA.}%
\thanks{R. Gnanasambandam is with the Department of Industrial and Manufacturing Engineering, Florida A\&M University--Florida State University, Tallahassee, FL 32310 USA.}%
\thanks{J. Yu is with the DEVCOM, Army Research Laboratory, Adelphi, MD 20783 USA.}%
\thanks{Corresponding author: Zhenyu James Kong (e-mail: zkong@vt.edu).}}

\begin{document}
\maketitle

\begin{abstract}
Additive friction stir deposition (AFSD) involves coupled thermal and material-flow fields governed by frictional heating, plastic deformation, and tool-imposed boundary conditions. High-fidelity finite-volume methods (FVMs) can resolve these fields but are computationally expensive, while constant-property approximations can distort the predicted thermo-mechanical response. This work develops a steady-state multi-task coupled physics-informed neural network (MCoPINN) to predict three-dimensional temperature and material-flow fields while reconstructing temperature-dependent thermal conductivity, density, and specific heat from sparse property measurements. A theoretical analysis decomposes prediction error into property reconstruction and neural field-solver contributions and characterizes the discrepancy associated with constant-property models. A controlled one-dimensional nonlinear-conduction study evaluates property reconstruction under sparse measurements. The framework is then evaluated for AFSD against a temperature-dependent FVM benchmark and experimental thermocouple measurements. MCoPINN reproduces the reference thermal and material-flow fields and reduces average temperature root-mean-square error from 50.13~K for the constant-property CoPINN to 30.28~K across four traverse speeds. MCoPINN requires approximately 8.5~hours of training per condition, compared with 52~hours for the FVM benchmark. These results demonstrate computationally efficient, physics-coupled prediction with temperature-dependent material properties.
\end{abstract}


\noindent\textit{Note to Practitioners---}
AFSD process models are useful for selecting process conditions and interpreting thermal and material-flow behavior, but high-fidelity simulations can require many hours for each operating condition. The proposed framework provides a lower-cost alternative that retains the governing physics and can also reconstruct temperature-dependent material properties when only sparse property measurements are available. For materials with well-established property laws, those laws should be prescribed directly. The proposed property-learning component is intended for partially characterized materials, where the continuous temperature dependence is uncertain or incompletely measured. The present study considers steady-state, single-track AFSD. Extension to transient and multilayer deposition is required before deployment for full build-scale process planning.

\section{Introduction}
\label{sec:introduction}
\IEEEPARstart{A}{dditive} friction stir deposition (AFSD) is a solid-state metal additive manufacturing process in which feedstock is deposited through severe plastic deformation and frictional heating without melting. Relative to fusion-based routes, AFSD can reduce solidification-related defects and residual stress while promoting favorable microstructural evolution and mechanical properties \cite{Yu01022021,SHAO2023112356,AGRAWAL2023112470,KUMAR2026101243}. These advantages have motivated the use of AFSD for cladding, repair, and restoration of structural alloys \cite{stubblefield2023ballistic,palya2026investigation,HARTLEY2021117045}. Process modeling of AFSD remains challenging because accurate prediction of the thermo-mechanical response requires simultaneous treatment of temperature, deformation, constitutive behavior, heat generation, and material flow near the tool--material interface. Consequently, numerous studies have been conducted in this direction using high-fidelity numerical methods, data-driven approaches, and physics-informed modeling techniques, as reviewed below.

\subsection{Literature Review}
\label{subsec:literature_review}

\subsubsection{High-Fidelity Numerical Simulations}
\label{subsubsec:high_fidelity}

Finite-element methods (FEMs), finite-volume methods (FVMs), and mesh-free formulations such as smoothed particle hydrodynamics (SPH) have been used extensively to predict the coupled thermal and mechanical response of AFSD~\cite{SHARMA2026100220}. These methods can resolve temperature, deformation, stress, and material-flow fields, but at a substantial computational cost. Alam et al.~\cite{ALAM2025100263} developed a coupled Eulerian--Lagrangian FEM model for AA6061-T6 that resolves large deformation, contact, material flow, and heat transfer. In contrast, Sharma et al.~\cite{SHARMA2022111412} proposed a computationally efficient pseudo-thermo-mechanical FEM framework for multilayer AFSD. However, material flow is not solved directly, and the strain rate is estimated analytically to determine heat generation.

FVM and computational fluid dynamics formulations have been used to resolve strongly deforming material flow more directly. Kincaid et al.~\cite{10.1115/1.4056642} developed a viscoplastic OpenFOAM model for single-layer AA6061 deposition that reproduced the measured temperature and deposit morphology. Gotawala and Yu~\cite{GOTAWALA2023114} developed an experiment-informed FVM framework coupling temperature, effective strain rate, constitutive response, and material flow. The corresponding FVM benchmark requires approximately 52~hours per operating condition because the coupled equations must be solved iteratively to convergence. Fernando et al.~\cite{fernando2025toolgeometry} and Habib et al.~\cite{habib2026pingeometry} further investigated the influence of tool-face and pin geometries on material flow, thermo-mechanical behavior, and interlayer material mixing.

Mesh-free SPH avoids severe mesh distortion and can track material histories. Stubblefield et al.~\cite{STUBBLEFIELD2021109514} used SPH to resolve material flow, temperature, and strain-rate gradients, and later combined particle tracking with X-ray computed tomography to validate material transport around the tool~\cite{stubblefield2023materialflow}. Reported single-track SPH simulations required more than 35~hours even with GPU acceleration. Stubblefield et al.~\cite{stubblefield2022mechanics} also demonstrated that temperature- and strain-rate-dependent constitutive behavior substantially affects the predicted temperature, stress, and deformation fields, highlighting the importance of accurate material-property representations.

Reduced thermal and analytical formulations can lower computational cost by simplifying heat generation and material flow. Hu et al.~\cite{chen2024thermalAFSD} developed separate models for friction preheating and steady deposition and reproduced measured thermal histories. Ruvo et al.~\cite{ruvo2026surfaceheat} proposed a surface-based heat-source model for multilayer AFSD coupled to an elastoplastic mechanical model. Related studies used calibrated heat-source representations to reconstruct thermal histories and microstructural evolution~\cite{lyu2025genetic}, while Kocak et al.~\cite{kocak2026asymptotes} examined thermo-mechanical evolution across successive deposition passes. An analytical model by de Vaucorbeil et al.~\cite{devaucorbeil2026analytical} related process conditions to processing temperature, deposition rate, and energy efficiency. These formulations provide efficient alternatives for thermal prediction and process analysis but do not resolve the complete coupled material-flow response.

\subsubsection{Data-Driven and Physics-Informed Learning-Based Approaches}
\label{subsubsec:data_physics_informed}

Data-driven methods provide another route to lower-cost prediction~\cite{MABUWA2026123154}. Zhu et al.~\cite{ZHU20241} combined a physics-based thermal model, in-situ temperature measurements, and a Gaussian-process surrogate to estimate uncertain parameters and accelerate prediction. Other studies have used machine learning to relate process conditions to material properties. Qiao et al.~\cite{qiao2024comparative} compared support vector machines, random forests, and neural networks for predicting microhardness and tensile strength. Similarly, Zhu et al.~\cite{zhu2024mlclassification} and Tam et al.~\cite{tam2026mechanicalML} investigated machine-learning-based prediction of tensile properties from process parameters and experimental data. Although computationally efficient, these approaches generally predict selected quantities rather than complete thermal and material-flow fields. Their accuracy also depends on sufficient training data, which can be costly to obtain.

Machine learning has also been applied to AFSD process monitoring and physics-assisted prediction. Zhang et al.~\cite{zhang2026forceflow} used force-signal features for online defect detection and classification. Shi et al.~\cite{shi2024afsdphysics} developed AFSD-Physics to identify relations governing tool and build temperatures using physical information and measurements. They subsequently introduced AFSD-Nets~\cite{shi2024afsdnets} for predicting coupled temperature evolution during multilayer deposition. Wu et al.~\cite{wu2026sequentialbayesian} combined three-dimensional computational fluid dynamics with thermal imaging in a sequential Bayesian framework for full-field temperature and material-flow prediction. However, these approaches often require experimental measurements or precomputed simulations that can be expensive to generate.

Physics-informed neural networks (PINNs) embed governing equations and boundary conditions directly into training~\cite{RAISSI2019686,712178,cuomo2022scientific}. Gnanasambandam et al.~\cite{11305145} introduced a coupled PINN (CoPINN) for steady-state AFSD that predicts full-field velocity and temperature at substantially lower computational cost than high-fidelity simulations. However, the formulation assumes constant thermophysical properties. Outside AFSD, recent physics-informed studies have identified temperature-dependent material properties from sparse field measurements~\cite{ZHENG2025104408,Moon2025}. These studies infer continuous property functions indirectly from measured field responses, requiring dedicated experiments that can be costly to perform.

\subsection{Research Gap and Contributions}
\label{subsec:research_gap}
The literature review reveals two related limitations. First, several AFSD models, including the CoPINN developed by Gnanasambandam et al.~\cite{11305145}, assume constant thermal conductivity, density, and specific heat, although these properties can vary appreciably with temperature. Accounting for this temperature dependence in high-fidelity simulations can further increase the computational burden. In addition, continuous property relations are not always fully characterized. Material properties are often reported only at a limited number of temperatures. When the functional form is known and simple, such as a low-order polynomial, a few measurements may be sufficient to determine the relation. However, when the underlying relation is unknown or strongly nonlinear, sparse measurements may not adequately capture its behavior. Conventional interpolation may also fail to reconstruct the relation accurately. These challenges motivate a computationally efficient, physics-coupled approach for reconstructing temperature-dependent material properties from limited data. This sparse-data setting is further examined in a controlled one-dimensional (1D) study.

Second, existing PINN-based formulations for AFSD have primarily focused on reducing the computational cost of field prediction. They do not explicitly distinguish errors arising from the material-property representation from those introduced by the neural field solver. This distinction is important because a small governing-equation residual does not necessarily imply an accurate physical prediction when the underlying material-property representation is inaccurate. For example, the previously developed CoPINN~\cite{11305145} predicts the coupled AFSD fields efficiently but assumes constant thermophysical properties throughout the domain.

The present study addresses these limitations through two primary contributions. First, a multi-task coupled physics-informed neural network (MCoPINN) augments the thermal and material-flow subnetworks with a property subnetwork that reconstructs temperature-dependent thermophysical properties from discrete measurements while remaining coupled to the governing physics. Second, a formal theoretical analysis separates the total prediction error into two sources: error from reconstructing the thermophysical properties and error from the PINN field solution. The analysis also examines the additional discrepancy introduced when temperature-dependent properties are assumed to be constant.

The remainder of the paper is organized as follows. Section~\ref{sec:methodology} presents the model formulation, and Section~\ref{sec:theory} develops the theoretical analysis. Section~\ref{sec:results_discussion} presents the results and performance evaluation, followed by conclusions in Section~\ref{sec:conclusion}.

\section{Methodology}
\label{sec:methodology}

\subsection{Experimental Setup}
\label{subsec:afsd_configuration}
In AFSD, metallic feedstock, such as AA6061 and AA7075 aluminum alloys, is supplied through a rotating tool and deposited onto a substrate by frictional heating and severe plastic deformation. Its potential for large-scale fabrication and refined microstructures has attracted growing interest. Figure~\ref{fig:afsd_process} illustrates the experimental configuration. The tool rotates with angular velocity $\omega$, traverses with velocity $V$, and receives a normal feed rate $F$. A thermocouple positioned within the print head provides temperature measurements near the tool--material interface. These measurements are used only for independent validation of the numerical predictions.

\begin{figure}[ht]
\centering
\includegraphics[width=0.6\columnwidth]{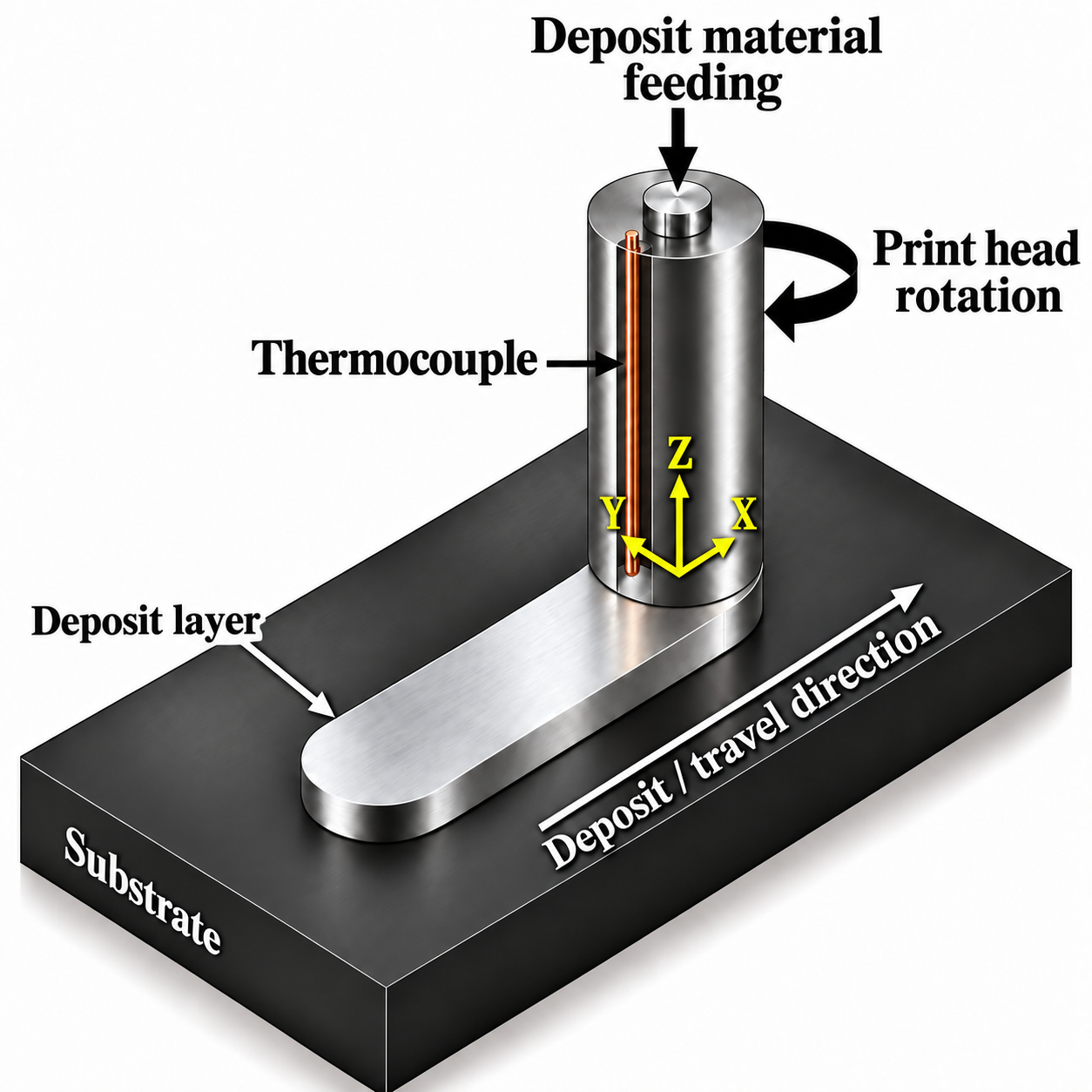}
\caption{Schematic of the AFSD process.}
\label{fig:afsd_process}
\end{figure}

\subsection{Governing Equations}
\label{subsec:governing_model}
The governing equations follow the steady-state AFSD formulation developed by Gotawala and Yu~\cite{GOTAWALA2023114} and are summarized below. Here, steady-state refers to temperature and material-flow fields that remain time-independent in a reference frame moving with the tool.
\subsubsection{Conservation Laws}
\label{subsubsec:conservation_laws}
The deposited material is modeled as a steady viscoplastic continuum in a three-dimensional domain $\Omega\subset\mathbb R^3$. The state variables are material-flow velocity $\mathbf u=(u_1,u_2,u_3)$, pressure $p$, and temperature $T$, where $u_1$, $u_2$, and $u_3$ represent the velocity components along the three spatial directions, respectively. The governing partial differential equations (PDEs) are given by
\begin{subequations}\label{eq:governing}
\begin{gather}
\frac{\partial[\rho u_i]}{\partial x_i}=0,
\label{eq:mass_conservation}
\\
\frac{\partial[\rho u_i u_j]}{\partial x_i}
=
-\frac{\partial p}{\partial x_j}
+\frac{\partial}{\partial x_i}
\left[
\mu_{\mathrm{vis}}
\left(
\frac{\partial u_j}{\partial x_i}
+\frac{\partial u_i}{\partial x_j}
\right)
\right],
\label{eq:momentum_conservation}
\\
\rho c_p u_i\frac{\partial T}{\partial x_i}
=
\frac{\partial}{\partial x_i}
\left(
k\frac{\partial T}{\partial x_i}
\right)
+q_g.
\label{eq:energy_conservation}
\end{gather}
\end{subequations}
Equations~\eqref{eq:mass_conservation}--\eqref{eq:energy_conservation} enforce conservation of mass, momentum, and energy, respectively. Here, $x_i$ denotes the spatial coordinates, with $i,j\in\{1,2,3\}$. Density $\rho(T)$, thermal conductivity $k(T)$, and specific heat capacity $c_p(T)$ are temperature dependent. The quantities $\mu_{\mathrm{vis}}$ and $q_g$ denote the effective viscosity and volumetric heat generation due to plastic deformation, respectively.

\subsubsection{Constitutive Relations}
\label{subsubsec:constitutive_relations}
The governing PDEs require constitutive relations to relate deformation to stress, viscosity, and heat generation. The constitutive response follows the Sellars--Tegart formulation and is summarized as follows:
\begin{subequations}\label{eq:constitutive}
\begin{gather}
D_{ij}
=
\frac{1}{2}
\left(
\frac{\partial u_i}{\partial x_j}
+
\frac{\partial u_j}{\partial x_i}
\right),
\label{eq:strain_tensor}
\\
\dot{\varepsilon}_{\mathrm{eff}}
=
\left(
\frac{2}{3}D_{ij}D_{ij}
\right)^{1/2},
\label{eq:effective_strain_rate}
\\
Z
=
\dot{\varepsilon}_{\mathrm{eff}}
\exp\!\left(
\frac{E_a}{R_gT}
\right),
\label{eq:zener_hollomon}
\\
\sigma_{\mathrm{eff}}
=
\frac{1}{\alpha_\sigma}
\sinh^{-1}\!\left[
\left(
\frac{Z}{A}
\right)^{1/n}
\right],
\label{eq:effective_stress}
\\
\mu_{\mathrm{vis}}
=
\frac{\sigma_{\mathrm{eff}}}
{3\dot{\varepsilon}_{\mathrm{eff}}},
\label{eq:effective_viscosity}
\\
q_g
=
\chi\sigma_{\mathrm{eff}}
\dot{\varepsilon}_{\mathrm{eff}}.
\label{eq:volumetric_heat}
\end{gather}
\end{subequations}
Equations~\eqref{eq:strain_tensor} and~\eqref{eq:effective_strain_rate} define the strain-rate tensor $D_{ij}$ and effective strain rate $\dot{\varepsilon}_{\mathrm{eff}}$, respectively. Equation~\eqref{eq:zener_hollomon} defines the Zener--Hollomon parameter $Z$, where $E_a$ is the activation energy and $R_g$ is the universal gas constant. Equation~\eqref{eq:effective_stress} gives the effective stress $\sigma_{\mathrm{eff}}$, where $\alpha_\sigma$, $A$, and $n$ are Sellars--Tegart material parameters. Equations~\eqref{eq:effective_viscosity} and~\eqref{eq:volumetric_heat} define the effective viscosity and volumetric heat generation, respectively, where $\chi$ denotes the fraction of plastic work converted into heat. These relations provide the thermo-mechanical coupling, since temperature and material-flow velocity jointly determine the effective stress, viscosity, and plastic heat generation.

\subsubsection{Boundary Conditions}
\label{subsubsec:boundary_conditions}
The top surface is divided into distinct boundary regions associated with the feed rod, print head shoulder, and exposed material, as illustrated in Figure~\ref{fig:top_boundary}. BC~2 corresponds to the central feed-rod region, while BC~1 corresponds to the surrounding print head shoulder region.

\begin{figure}[!t]
\centering
\includegraphics[width=0.5\columnwidth]{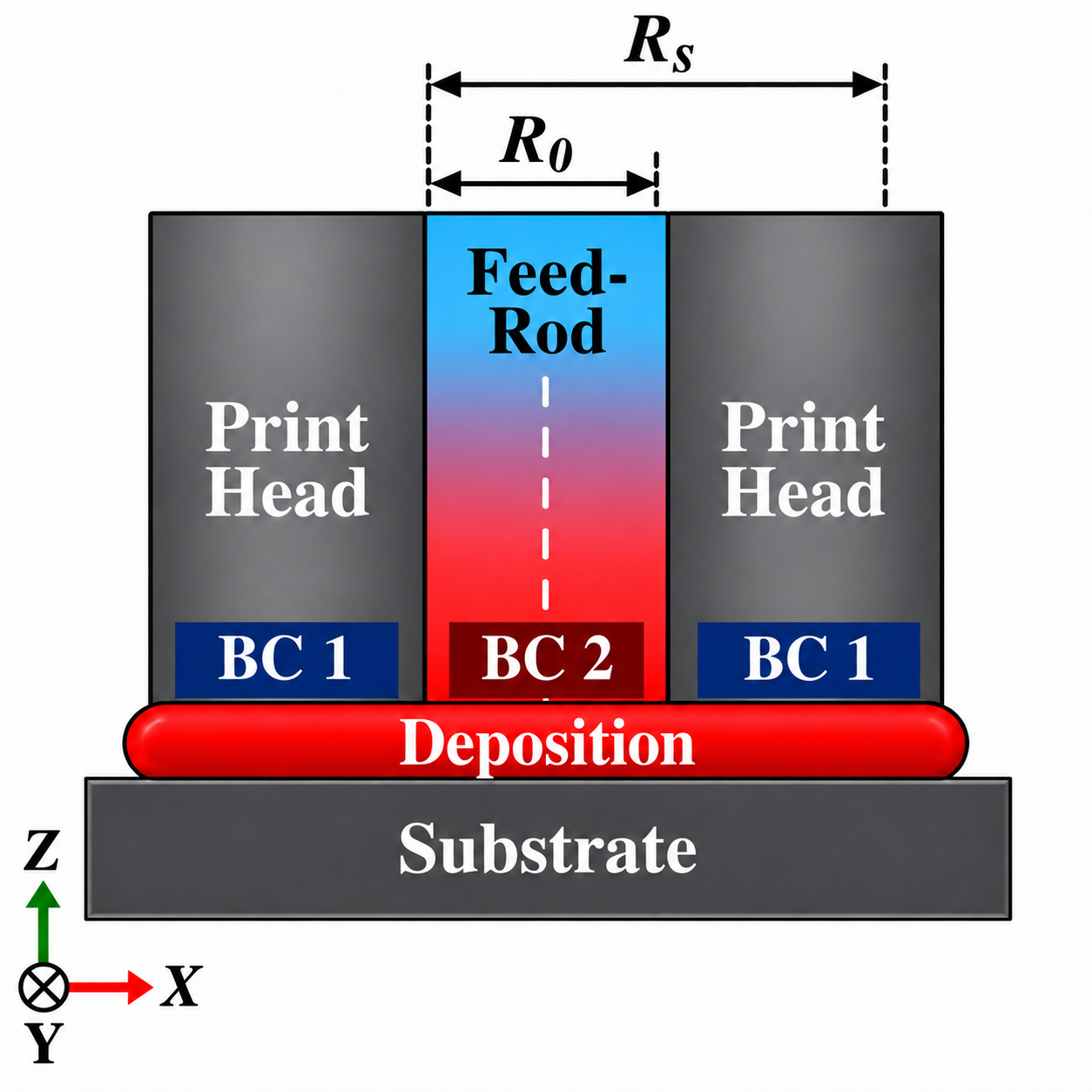}
\caption{Top-surface boundary regions beneath the AFSD tool.}
\label{fig:top_boundary}
\end{figure}

The radial distance from the tool center $(x_c,y_c)$ is defined as
$r=[(x-x_c)^2+(y-y_c)^2]^{1/2}$, with
$x-x_c=r\cos\theta$ and $y-y_c=r\sin\theta$, where $\theta$ is the
circumferential angle measured from the positive $x$-direction. Let $R_0$
and $R_s$ denote the feed-rod and shoulder radii, respectively. The
feed-rod region is assumed to be fully sticking, whereas partial slip is
prescribed over the shoulder through
$\delta=\beta_s(r-R_0)$, where $\beta_s$ is the slip coefficient.

The tool-contact velocity boundary conditions are imposed as
\begin{equation}
\label{eq:velocity_bc}
\begin{aligned}
u_1&=
\begin{cases}
\omega r\sin\theta-V,
& r\le R_0 \quad \text{(BC~2)},\\
(1-\delta)\omega r\sin\theta-V,
& R_0<r\le R_s \quad \text{(BC~1)},
\end{cases}
\\[3pt]
u_2&=
\begin{cases}
-\omega r\cos\theta,
& r\le R_0 \quad \text{(BC~2)},\\
-(1-\delta)\omega r\cos\theta,
& R_0<r\le R_s \quad \text{(BC~1)},
\end{cases}
\\[3pt]
u_3&=
\begin{cases}
-F,
& r\le R_0 \quad \text{(BC~2)},\\
0,
& R_0<r\le R_s \quad \text{(BC~1)}.
\end{cases}
\end{aligned}
\end{equation}
Here, $\omega$ is the tool angular velocity, $V$ is the traverse velocity,
and $F$ is the feed rate. Equation~\eqref{eq:velocity_bc} imposes full
sticking over the feed-rod region and partial slip over the shoulder while
accounting for tool rotation, traverse motion, and feed motion. Outside the
tool-contact region, $(u_1,u_2,u_3)=(-V,0,0)$ is imposed.

The thermal boundary conditions on the top surface are
\begin{equation}
\label{eq:top_thermal_bc}
-k\nabla T\!\cdot\!\mathbf n=
\begin{cases}
-q_2,&r\le R_0 \quad \text{(BC~2)},\\
-q_1,&R_0<r\le R_s \quad \text{(BC~1)},\\
h(T-T_a),&r>R_s,
\end{cases}
\end{equation}
where $\mathbf n$ is the outward unit normal, $h$ is the convective heat-transfer coefficient, and $T_a$ is the ambient temperature. The heat fluxes $q_1$ and $q_2$ over the shoulder and feed-rod regions, respectively, are
\begin{align}
q_1&=\eta\left[
\chi(1-\delta)\frac{\sigma_{\mathrm{eff}}}{\sqrt{3}}
+\delta\mu_f\sigma_{\mathrm{eff}}
\right]\omega r,
\label{eq:shoulder_heat_flux}\\
q_2&=\lambda\chi\frac{\sigma_{\mathrm{eff}}}{\sqrt{3}}\omega r,
\label{eq:feedrod_heat_flux}
\end{align}
where $\mu_f$ is the friction coefficient and $\eta$ and $\lambda$ are heat-partition parameters.

The exposed side surfaces are subjected to the same convective heat-loss condition as the exposed top surface. At the bottom surface, heat loss is modeled as
\begin{equation}
\label{eq:bottom_thermal_bc}
-k\nabla T\cdot\mathbf n
=
C(T-T_a)^3,
\end{equation}
where $C$ is the nonlinear heat-loss coefficient. 

The numerical values of all the constants are provided in Section~\ref{subsec:mcopinn_afsd}.

\subsection{Proposed MCoPINN Formulation}
\label{subsec:mcopinn_formulation}
Figure~\ref{fig:mcopinn_architecture} illustrates the proposed MCoPINN architecture, which extends the CoPINN framework of Gnanasambandam et al.~\cite{11305145} by incorporating a property subnetwork. The framework consists of separate flow, thermal, and property subnetworks. For a spatial coordinate vector $\mathbf x=(x_1,x_2,x_3)$, the flow subnetwork predicts the material-flow velocity components $(u_1,u_2,u_3)$ and pressure $p$, while the thermal subnetwork predicts the temperature $T$. The predicted temperature is then supplied to the property subnetwork, which reconstructs the temperature-dependent thermal conductivity $\widehat{k}(T)$, density $\widehat{\rho}(T)$, and specific heat capacity $\widehat{c}_p(T)$. These mappings are expressed as
\begin{equation}
\label{eq:networks}
\begin{aligned}
\relax[\mathbf u,p]&=\mathcal N_f(\mathbf x;\boldsymbol{\theta}_f),\\
T&=\mathcal N_T(\mathbf x;\boldsymbol{\theta}_T),\\
\relax[\widehat{\rho},\widehat{c}_p,\widehat{k}]
&=\mathcal P(T;\boldsymbol{\psi}),
\end{aligned}
\end{equation}
where $\mathcal N_f$ and $\mathcal N_T$ denote the flow and thermal subnetworks, respectively, and $\mathcal P$ denotes the property subnetwork. The trainable parameters of these subnetworks are represented by $\boldsymbol{\theta}_f$, $\boldsymbol{\theta}_T$, and $\boldsymbol{\psi}$.

The reconstructed thermophysical properties, expressed as functions of temperature, are used in the governing equations in place of constant property values. This creates a two-way coupling within the MCoPINN framework. The predicted temperature determines the thermophysical properties, while the reconstructed properties influence the thermal and material-flow solutions through the governing physics.

\begin{figure}[!t]
\centering
\includegraphics[width=\columnwidth]{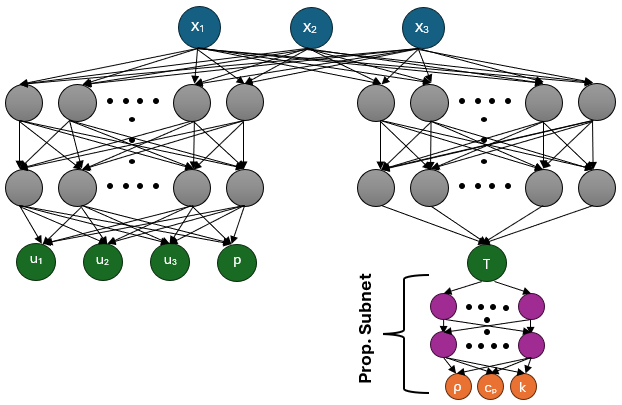}
\caption{Proposed MCoPINN architecture for coupled AFSD field prediction and reconstruction of temperature-dependent thermophysical properties.}
\label{fig:mcopinn_architecture}
\end{figure}

By virtue of this coupling, the flow, thermal, and property subnetworks are trained jointly. The spatial inputs are normalized to $[-1,1]$, and the PDE and boundary-condition residuals are evaluated using nondimensionalized forms of Equations~\eqref{eq:mass_conservation}--\eqref{eq:energy_conservation} and Equations~\eqref{eq:velocity_bc}--\eqref{eq:bottom_thermal_bc}, respectively. Available thermophysical property measurements at discrete temperature values are additionally used to constrain the property subnetwork. The total loss combines PDE, boundary-condition, and property-data contributions and is defined as
\begin{equation}
\label{eq:total_loss}
\begin{aligned}
\mathcal L={}&
w_{\mathrm{PDE}}
\sum_{j\in\mathcal I_{\mathrm{PDE}}}
\frac{1}{N_j^{\mathrm{PDE}}}
\sum_{n=1}^{N_j^{\mathrm{PDE}}}
\left|\mathcal R_j(\mathbf x_n)\right|^2
\\
&+
w_{\mathrm{BC}}
\sum_{j\in\mathcal I_{\mathrm{BC}}}
\frac{1}{N_j^{\mathrm{BC}}}
\sum_{n=1}^{N_j^{\mathrm{BC}}}
\left|\mathcal B_j(\mathbf x_n)\right|^2
\\
&+
w_{\mathrm{prop}}
\sum_{P\in\{\rho,c_p,k\}}
\frac{1}{N_P}
\sum_{m=1}^{N_P}
\left|\widehat P(T_m)-P_m\right|^2 .
\end{aligned}
\end{equation}

Here, $\mathcal I_{\mathrm{PDE}}$ and $\mathcal I_{\mathrm{BC}}$ denote the sets of governing-equation and boundary-condition residuals, respectively. The quantities $\mathcal R_j$ and $\mathcal B_j$ denote the corresponding residuals, while $N_j^{\mathrm{PDE}}$ and $N_j^{\mathrm{BC}}$ are the associated numbers of collocation points. For each thermophysical property $P\in\{\rho,c_p,k\}$, $N_P$ is the number of available measurements, $P_m$ is the measured value at temperature $T_m$, and $\widehat P(T_m)$ is the corresponding prediction. The weights $w_{\mathrm{PDE}}$, $w_{\mathrm{BC}}$, and $w_{\mathrm{prop}}$ control the relative contributions of the three loss components.

\section{Theoretical Analysis of Property Sensitivity and MCoPINN Prediction Error}
\label{sec:theory}
This section presents three main theoretical results to explain how different sources of error affect model predictions. Theorem~\ref{thm:solution_map} shows how errors in reconstructed material properties affect the predicted fields. Theorem~\ref{thm:mcopinn_error} decomposes the total prediction error into contributions from property reconstruction and the PINN field solution. Theorem~\ref{thm:constant_separation} shows that assuming constant material properties can introduce errors even when the governing equations are solved accurately, under the stated assumptions. The controlled 1D study and the AFSD application in Section~\ref{sec:results_discussion} provide numerical evidence consistent with these theoretical results.\\
Let $I=[T_{\min},T_{\max}]$ denote the admissible temperature interval, with $\tau\in I$ denoting temperature. The thermophysical-property vector $\mathbf P$ and coupled state $\mathbf z$ are defined as
\begin{equation}
\label{eq:analysis_state}
\begin{aligned}
\mathbf P(\tau)
&=
\bigl(\rho(\tau),c_p(\tau),k(\tau)\bigr),\\
\mathbf z(\mathbf x)
&=
\bigl(\mathbf u(\mathbf x),p(\mathbf x),T(\mathbf x)\bigr).
\end{aligned}
\end{equation}
Here, $\mathbf P\in\mathcal A$, where $\mathcal A$ denotes the admissible thermophysical-property function space, and $\mathbf z\in\mathcal X$, where $\mathcal X$ denotes the coupled state space. The governing PDEs and boundary conditions are collected in the residual operator
\begin{equation}
\label{eq:residual_operator}
R:\mathcal A\times\mathcal X\rightarrow\mathcal Y,
\qquad
R(\mathbf P,\mathbf z)=0.
\end{equation}
For prescribed properties $\mathbf P$, the corresponding exact solution set is
\begin{equation}
\label{eq:solution_set}
\mathcal Z(\mathbf P)
=
\left\{
\mathbf z\in\mathcal X:
R(\mathbf P,\mathbf z)=0
\right\}.
\end{equation}
Let $\mathbf P^\star$ and $\mathbf z^\star$ denote the reference properties and corresponding solution, such that
\begin{equation}
\label{eq:reference_solution}
R(\mathbf P^\star,\mathbf z^\star)=0.
\end{equation}
For the controlled 1D problem, $\mathbf P^\star$ is known exactly, whereas for AFSD it denotes the prescribed AA6061 property relations used as the numerical reference. The MCoPINN predictions are denoted by $\widehat{\mathbf P}_M$ and $\widehat{\mathbf z}_M$. Table~\ref{tab:analysis_notation} summarizes the notation used in the theoretical analysis.
\begin{table*}[!t]
\caption{Notation Used in the Theoretical Analysis}
\label{tab:analysis_notation}
\centering
\footnotesize
\setlength{\tabcolsep}{5pt}
\renewcommand{\arraystretch}{1.12}
\begin{tabularx}{\textwidth}{
>{\centering\arraybackslash}p{0.09\textwidth}
>{\raggedright\arraybackslash}X
>{\centering\arraybackslash}p{0.09\textwidth}
>{\raggedright\arraybackslash}X}
\toprule
\textbf{Symbol} & \textbf{Meaning}
& \textbf{Symbol} & \textbf{Meaning} \\
\midrule

$I,\tau$
& Admissible temperature interval and temperature argument
&
$\mathcal A$
& Banach space of thermophysical-property functions
\\

$\mathcal X$
& Sobolev product state space
&
$\mathcal Y$
& Banach product space of governing-equation and boundary-condition residuals
\\

$\mathbf P(\tau)$
& Thermophysical-property vector $(\rho,c_p,k)$
&
$\mathbf z(\mathbf x)$
& Coupled state $(\mathbf u,p,T)$
\\

$\mathbf P^\star,\mathbf z^\star$
& Reference properties and corresponding solution
&
$\mathcal Z(\mathbf P)$
& Exact solution set associated with $\mathbf P$
\\

$\widehat{\mathbf P}_M,\widehat{\mathbf z}_M$
& MCoPINN reconstructed properties and predicted state
&
$\mathbf P_C,\mathcal Z_C$
& Constant-property representation and exact solution set
$\mathcal Z_C=\mathcal Z(\mathbf P_C)$
\\

$\mathbf z_C,\widehat{\mathbf z}_C$
& Exact and predicted constant-property states
&
$\eta_C$
& Positive residual-separation constant for the constant-property solution set
\\

$d_C$
& Distance between $\mathbf z^\star$ and $\mathcal Z_C$
&
$\varepsilon_C$
& Constant-property solver error
\\

$\mathcal U_{\mathcal A},\mathcal U_{\mathcal X}$
& Open convex neighborhoods of $\mathbf P^\star$ and $\mathbf z^\star$
&
$\mathcal U_M$
& Sufficiently small neighborhood of
$(\mathbf P^\star,\mathbf z^\star)$ contained in
$\mathcal U_{\mathcal A}\times\mathcal U_{\mathcal X}$
\\

$C^1$
& Class of continuously Fr\'echet differentiable mappings
&
$DR$
& Derivative of the residual operator $R$
\\

$D_{\mathbf P}R,D_{\mathbf z}R$
& Derivatives of $R$ with respect to properties and state
&
$D\mathcal S$
& Derivative of the local solution map
\\

$\mathcal J^\star$
& State derivative $D_{\mathbf z}R(\mathbf P^\star,\mathbf z^\star)$
&
$M_z$
& Bound on $\|(\mathcal J^\star)^{-1}\|$
\\

$\Delta\mathbf P$
& Perturbation of the property functions
&
$\mathcal S(\mathbf P)$
& Local solution map associated with $\mathbf P$
\\

$C_S$
& Local property-to-state sensitivity constant
&
$C_R$
& Local residual-to-state stability constant
\\

$\widehat{\varepsilon}_R$
& Empirical residual evaluated at training points
&
$\varepsilon_{\mathrm{samp}}$
& Sampling discrepancy between discrete and continuous residuals
\\

$\varepsilon_P$
& property reconstruction error
&
&
\\

\bottomrule
\end{tabularx}
\end{table*}

\subsection{Assumptions}
\label{subsec:theory_assumptions}
The analysis considers positive and bounded thermophysical properties over the operating temperature range, consistent with physical behavior and available material data. Let $(\mathbf P^\star,\mathbf z^\star)$ denote a reference solution, where $\mathbf P^\star$ represents the true thermophysical-property functions and $\mathbf z^\star$ the corresponding temperature and material-flow fields. The condition $R(\mathbf P^\star,\mathbf z^\star)=0$ indicates that this solution satisfies the governing equations and boundary conditions. The following assumptions support the theoretical proofs without limiting the general applicability of MCoPINN. They concern locally stable solutions under physically reasonable conditions.

\noindent\textbf{Assumption 1 (Local regularity and stability).}
There exist open convex neighborhoods $\mathcal U_{\mathcal A}$ and $\mathcal U_{\mathcal X}$ of $\mathbf P^\star$ and $\mathbf z^\star$, respectively, such that
\begin{equation}
R\in C^1\left(\mathcal U_{\mathcal A}\times\mathcal U_{\mathcal X};\mathcal Y\right),
\end{equation}
with locally Lipschitz derivative $DR$. Furthermore, the state derivative $\mathcal J^\star=D_{\mathbf z}R(\mathbf P^\star,\mathbf z^\star)$ has a bounded inverse:
\begin{equation}
\left\|(\mathcal J^\star)^{-1}\right\|\le M_z<\infty.
\end{equation}
The regularity condition is reasonable for smooth constitutive relations within the operating regime. Bounded invertibility corresponds to a locally well-posed and stable steady-state solution. This is a reasonable condition for well-posed thermo-mechanical models, including those used in AFSD and other metal additive manufacturing processes, within stable operating regimes.

\noindent\textbf{Assumption 2 (Locality of the MCoPINN prediction).}
Let $\mathcal U_M$ denote a sufficiently small neighborhood of $(\mathbf P^\star,\mathbf z^\star)$ in $\mathcal A\times\mathcal X$, with $\mathcal U_M\subset\mathcal U_{\mathcal A}\times\mathcal U_{\mathcal X}$. The MCoPINN prediction satisfies $(\widehat{\mathbf P}_M,\widehat{\mathbf z}_M)\in\mathcal U_M$. This condition is reasonable for sufficiently accurate predictions and is supported by the agreement with the reference solutions observed in the numerical studies. It specifies the local regime in which the prediction error bounds apply, without imposing any additional restriction on the MCoPINN architecture.

\noindent\textbf{Assumption 3 (Constant-property separation).}
For the constant-property representation $\mathbf P_C$, the solution set $\mathcal Z_C=\mathcal Z(\mathbf P_C)$ is nonempty and satisfies
\begin{equation}
\label{eq:assumption_a3}
\inf_{\mathbf z_C\in\mathcal Z_C}\left\|R(\mathbf P^\star,\mathbf z_C)\right\|_{\mathcal Y}\ge\eta_C>0.
\end{equation}
This condition is reasonable when temperature-dependent properties produce a coupled response that cannot be reproduced exactly by constant properties. The numerical comparisons demonstrate such differences, although strict separation may not hold universally. The condition is required only for establishing a positive constant-property model discrepancy and does not affect the main MCoPINN error bound.

\subsection{Property Sensitivity and Prediction Error}
\label{subsec:property_sensitivity}
\begin{theorem}[Local Property Sensitivity]
\label{thm:solution_map}
Under Assumption 1, there exists a locally unique continuously differentiable solution map
\begin{equation}
\label{eq:solution_map}
R\bigl(\mathbf P,\mathcal S(\mathbf P)\bigr)=0,
\qquad
\mathcal S(\mathbf P^\star)=\mathbf z^\star.
\end{equation}
For an admissible perturbation $\Delta\mathbf P$,
\begin{equation}
\label{eq:sensitivity}
D\mathcal S(\mathbf P^\star)[\Delta\mathbf P]
=
-(\mathcal J^\star)^{-1}
D_{\mathbf P}R(\mathbf P^\star,\mathbf z^\star)
[\Delta\mathbf P].
\end{equation}
If Assumption 2 also holds, then there exists a constant $C_S>0$ such that
\begin{equation}
\label{eq:property_bound}
\left\|
\mathcal S(\widehat{\mathbf P}_M)-\mathbf z^\star
\right\|_{\mathcal X}
\le
C_S\varepsilon_P,
\end{equation}
where
\begin{equation}
\label{eq:property_error}
\varepsilon_P
=
\left\|
\widehat{\mathbf P}_M-\mathbf P^\star
\right\|_{\mathcal A}.
\end{equation}
\end{theorem}

\begin{IEEEproof}
Under Assumption 1, the Banach-space implicit function theorem \cite[Theorem~15.1]{deimling1985nonlinear} applied to Equation~\eqref{eq:residual_operator} gives the local map $\mathcal S$. Differentiating
$R(\mathbf P,\mathcal S(\mathbf P))=0$
in the direction $\Delta\mathbf P$ gives
\begin{equation}
\label{eq:sensitivity_intermediate}
D_{\mathbf P}R[\Delta\mathbf P]
+
D_{\mathbf z}R
\left[
D\mathcal S[\Delta\mathbf P]
\right]
=0.
\end{equation}
Evaluating at $(\mathbf P^\star,\mathbf z^\star)$ and using the bounded invertibility condition in Assumption 1 yields Equation~\eqref{eq:sensitivity}.

Continuity of $D\mathcal S$ implies that $\|D\mathcal S\|\le C_S$ in a sufficiently small neighborhood. Applying the mean-value estimate between $\mathbf P^\star$ and $\widehat{\mathbf P}_M$ therefore gives
\begin{equation}
\left\|
\mathcal S(\widehat{\mathbf P}_M)-\mathcal S(\mathbf P^\star)
\right\|_{\mathcal X}
\le
C_S
\left\|
\widehat{\mathbf P}_M-\mathbf P^\star
\right\|_{\mathcal A},
\end{equation}
which yields Equation~\eqref{eq:property_bound}.
\end{IEEEproof}

Thus, property reconstruction errors propagate through the coupled solution and can affect both the thermal and material-flow fields. Lemma~\ref{lem:residual_state} complements this result by relating the governing-equation residual to the corresponding state error.

\begin{lemma}[Local Residual-to-State Stability]
\label{lem:residual_state}
Under Assumption 1, there exists $C_R>0$ such that,
\begin{equation}
\label{eq:residual_state}
\left\|
\mathbf z-\mathcal S(\mathbf P)
\right\|_{\mathcal X}
\le
C_R
\left\|
R(\mathbf P,\mathbf z)
\right\|_{\mathcal Y}.
\end{equation}
\end{lemma}

\begin{IEEEproof}
Let $\mathbf e=\mathbf z-\mathcal S(\mathbf P)$. Since
$R(\mathbf P,\mathcal S(\mathbf P))=0$, a Taylor expansion gives
\begin{equation}
\label{eq:residual_taylor}
R(\mathbf P,\mathbf z)
=
D_{\mathbf z}R
\bigl(\mathbf P,\mathcal S(\mathbf P)\bigr)[\mathbf e]
+
\mathbf r(\mathbf e),
\end{equation}
with $\|\mathbf r(\mathbf e)\|_{\mathcal Y}=O(\|\mathbf e\|_{\mathcal X}^2)$. By the local regularity and bounded invertibility conditions in Assumption 1, the state derivative remains uniformly invertible on a sufficiently small neighborhood. Therefore, within this neighborhood, the quadratic remainder can be absorbed, yielding Equation~\eqref{eq:residual_state}.
\end{IEEEproof}

PINN training evaluates residuals only at finite collocation and boundary points. Let $\widehat{\varepsilon}_R$ denote the empirical residual and $\varepsilon_{\mathrm{samp}}$ the discrepancy between this discrete measure and the continuous residual norm. Then
\begin{equation}
\label{eq:sampling}
\left\|
R(\widehat{\mathbf P}_M,\widehat{\mathbf z}_M)
\right\|_{\mathcal Y}
\le
\widehat{\varepsilon}_R
+
\varepsilon_{\mathrm{samp}}.
\end{equation}

\begin{theorem}[MCoPINN Error Decomposition]
\label{thm:mcopinn_error}
Under Assumptions 1 and 2, the MCoPINN prediction error is bounded by the combined contributions of the field-solver and property reconstruction errors, as follows.
\begin{equation}
\label{eq:mcopinn_bound}
\left\|
\widehat{\mathbf z}_M-\mathbf z^\star
\right\|_{\mathcal X}
\le
C_R
\left(
\widehat{\varepsilon}_R+\varepsilon_{\mathrm{samp}}
\right)
+
C_S\varepsilon_P.
\end{equation}
\end{theorem}

\begin{IEEEproof}
Add and subtract the exact state corresponding to the reconstructed properties:
\begin{equation}
\label{eq:error_split}
\widehat{\mathbf z}_M-\mathbf z^\star
=
\left[
\widehat{\mathbf z}_M-\mathcal S(\widehat{\mathbf P}_M)
\right]
+
\left[
\mathcal S(\widehat{\mathbf P}_M)-\mathbf z^\star
\right].
\end{equation}
The triangle inequality gives
\begin{equation}
\label{eq:error_triangle}
\begin{aligned}
\left\|
\widehat{\mathbf z}_M-\mathbf z^\star
\right\|_{\mathcal X}
\le{}&
\left\|
\widehat{\mathbf z}_M-\mathcal S(\widehat{\mathbf P}_M)
\right\|_{\mathcal X}\\
&+
\left\|
\mathcal S(\widehat{\mathbf P}_M)-\mathbf z^\star
\right\|_{\mathcal X}.
\end{aligned}
\end{equation}
Lemma~\ref{lem:residual_state} and Equation~\eqref{eq:sampling} bound the first term by
$C_R(\widehat{\varepsilon}_R+\varepsilon_{\mathrm{samp}})$.
Theorem~\ref{thm:solution_map} bounds the second by $C_S\varepsilon_P$. Their sum gives Equation~\eqref{eq:mcopinn_bound}.
\end{IEEEproof}

The first term in Equation~\eqref{eq:mcopinn_bound} represents the neural field-solver and sampling contribution, while the second represents property reconstruction error. When the reference property functions are prescribed directly, $\varepsilon_P=0$, isolating the field-solver contribution.

\subsection{Constant-Property Model Discrepancy}
\label{subsec:constant_property_discrepancy}

Because $\mathbf P_C$ need not be a small perturbation of $\mathbf P^\star$, the constant-property model is treated separately. Define
\begin{equation}
\label{eq:constant_distance}
d_C
=
\operatorname{dist}_{\mathcal X}
(\mathbf z^\star,\mathcal Z_C),
\end{equation}
and, for a constant-property numerical prediction $\widehat{\mathbf z}_C$,
\begin{equation}
\label{eq:constant_solver_error}
\varepsilon_C
=
\operatorname{dist}_{\mathcal X}
(\widehat{\mathbf z}_C,\mathcal Z_C).
\end{equation}
Here, $d_C$ represents model discrepancy and $\varepsilon_C$ represents solver error.

\begin{theorem}[Positive Constant-Property Model Discrepancy]
\label{thm:constant_separation}
Under Assumptions 1 and 3, the constant-property model has a positive discrepancy from the reference solution, as expressed below.
\begin{equation}
\label{eq:constant_model_gap}
d_C>0.
\end{equation}
\end{theorem}

\begin{IEEEproof}
Suppose, by contradiction, that $d_C=0$. Then there exists a sequence
$\{\mathbf z_{C,n}\}\subset\mathcal Z_C$ such that
\begin{equation}
\mathbf z_{C,n}\rightarrow\mathbf z^\star
\qquad\text{in }\mathcal X.
\end{equation}
For sufficiently large $n$, $(\mathbf P^\star,\mathbf z_{C,n})$ lies within the neighborhood specified in Assumption 1. By continuity of $R$,
\begin{equation}
R(\mathbf P^\star,\mathbf z_{C,n})
\rightarrow
R(\mathbf P^\star,\mathbf z^\star)
=0
\qquad\text{in }\mathcal Y.
\end{equation}
Therefore,
\begin{equation}
\inf_{\mathbf z_C\in\mathcal Z_C}
\left\|
R(\mathbf P^\star,\mathbf z_C)
\right\|_{\mathcal Y}
=0,
\end{equation}
which contradicts Assumption 3. Hence, $d_C>0$.
\end{IEEEproof}

\begin{corollary}[Prediction Error Floor]
\label{cor:constant_discrepancy}
Under Assumptions 1 and 3, the constant-property prediction error has a lower bound determined by the model discrepancy and solver error.
\begin{equation}
\label{eq:constant_total_bound}
\left\|
\widehat{\mathbf z}_C-\mathbf z^\star
\right\|_{\mathcal X}
\ge
d_C-\varepsilon_C.
\end{equation}
\end{corollary}

\begin{IEEEproof}
Since distance to a nonempty set is 1-Lipschitz,
\begin{equation}
d_C
\le
\left\|
\widehat{\mathbf z}_C-\mathbf z^\star
\right\|_{\mathcal X}
+
\varepsilon_C.
\end{equation}
Rearranging gives Equation~\eqref{eq:constant_total_bound}, while
Theorem~\ref{thm:constant_separation} establishes that $d_C>0$.
\end{IEEEproof}
For an accurately solved constant-property model, $\varepsilon_C$ is small relative to $d_C$, so the lower bound primarily reflects the model discrepancy introduced by the constant-property assumption rather than solver error. Combining Theorem~\ref{thm:mcopinn_error} and Corollary~\ref{cor:constant_discrepancy} gives the sufficient condition
\begin{equation}
\label{eq:comparison_condition}
C_R
\left(
\widehat{\varepsilon}_R+\varepsilon_{\mathrm{samp}}
\right)
+
C_S\varepsilon_P
<
d_C-\varepsilon_C.
\end{equation}
Under this condition,
\begin{equation}
\left\|
\widehat{\mathbf z}_M-\mathbf z^\star
\right\|_{\mathcal X}
<
\left\|
\widehat{\mathbf z}_C-\mathbf z^\star
\right\|_{\mathcal X},
\end{equation}
so the MCoPINN prediction is closer to the reference solution in the $\mathcal X$-norm. This condition provides a sufficient criterion for lower MCoPINN prediction error relative to the constant-property model, but it does not imply universal superiority of the learned-property formulation.

\section{Results and Discussion}
\label{sec:results_discussion}

\subsection{Case Study: 1D Nonlinear Conduction}
\label{subsec:controlled_1d}
A controlled 1D nonlinear-conduction problem is considered to isolate the effect of thermophysical-property reconstruction from the AFSD multiphysics. Figure~\ref{fig:1d_setup} illustrates the problem. A bar of length $L=0.01$~m is maintained at $T(0)=900$~K and $T(L)=300$~K. The thermal conductivity varies with temperature and is treated as the unknown property to be reconstructed.

\begin{figure}[H]
\centering
\includegraphics[width=0.9\columnwidth]{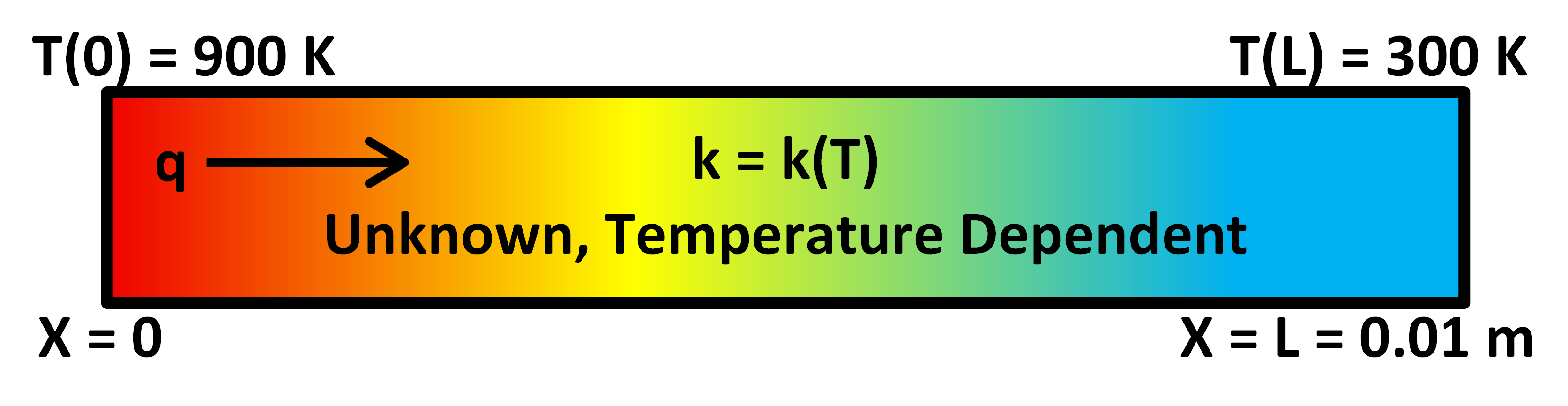}
\caption{1D nonlinear-conduction: problem setup.}
\label{fig:1d_setup}
\end{figure}

The governing heat-conduction equation and the prescribed ground-truth conductivity are
\begin{subequations}
\label{eq:1d_problem}
\begin{gather}
\frac{d}{dx}
\left[
k(T)\frac{dT}{dx}
\right]
=0,
\label{eq:1d_governing}
\\
k^\star(T)
=
11+
\frac{16}
{1+\exp[-(T-625)/30]}.
\label{eq:1d_true_k}
\end{gather}
\end{subequations}
Here, $k^\star(T)$ is expressed in W/(m$\cdot$K). The complete relation in Equation~\eqref{eq:1d_true_k} is hidden from B1--B5. It is used by B0 and to generate the reference FDM solution. This inverse setting follows standard temperature-dependent conductivity problems \cite{KIM2001537,Martin1D2000}.

\begin{figure}[H]
\centering
\includegraphics[width=0.4\columnwidth]{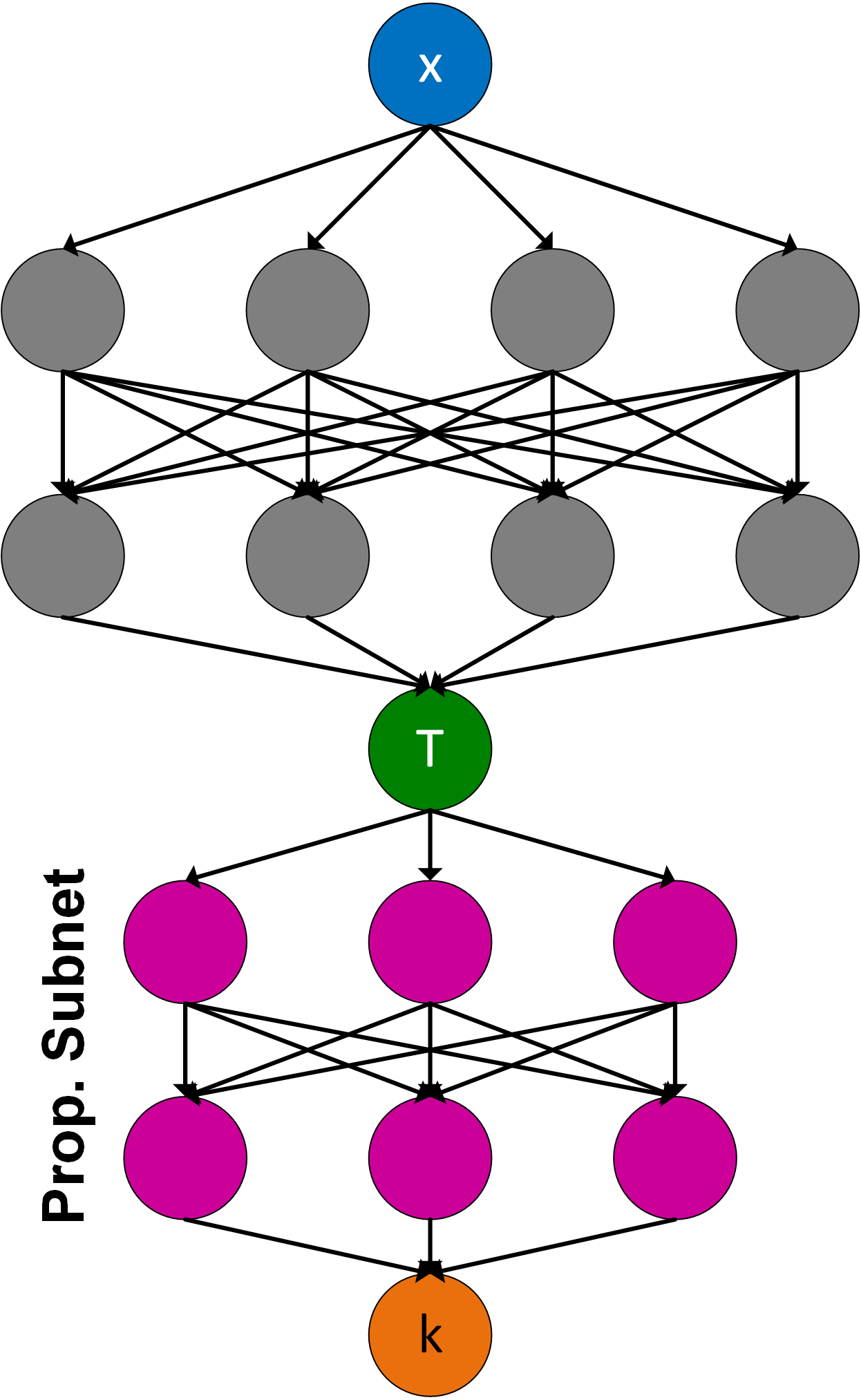}
\caption{Proposed PINN architecture for the 1D nonlinear-conduction study.}
\label{fig:1d_architecture}
\end{figure}

To compare different strategies for reconstructing the incomplete conductivity information, seven models are considered using the same temperature-PINN architecture and training setup. Table~\ref{tab:1d_models} summarizes these models. B0 uses the prescribed ground-truth conductivity function $k^\star(T)$ given in Equation~\eqref{eq:1d_true_k} and serves as the zero-property-error oracle. D0 assumes a constant conductivity of $20$~W/(m$\cdot$K). For B1--B5, only sparse samples of $k^\star(T)$ are provided. B1--B3 reconstruct the conductivity function using only these measurements through data-driven reconstruction methods, without physics coupling. B1 uses piecewise-linear interpolation, B2 uses cubic-spline interpolation, and B3 fits a simple neural network. Once reconstructed, the corresponding conductivity functions are kept fixed during temperature-PINN training. B4 uses a trainable cubic spline that is updated through the governing physics, while B5 uses the proposed property subnetwork.

\begin{table*}[!t]
\caption{Models Used in the Controlled 1D Study}
\label{tab:1d_models}
\centering
\footnotesize
\begin{tabular}{cllccllc}
\toprule
Model & Conductivity representation & Physics-coupled & Data &
Model & Conductivity representation & Physics-coupled & Data\\
\midrule
B0 & Ground-truth $k^\star(T)$ & No & Full law &
B3 & Data-only neural network & No & Sparse\\
D0 & Constant $k=20$ & No & Fixed &
B4 & Trainable cubic spline & Yes & Sparse\\
B1 & Piecewise-linear interpolation & No & Sparse &
B5 (proposed) & Property subnetwork & Yes & Sparse\\
B2 & Cubic-spline interpolation & No & Sparse &
& & &\\
\bottomrule
\end{tabular}
\end{table*}

Figure~\ref{fig:1d_architecture} illustrates the architecture used in the controlled 1D study. All models share the same temperature network, which takes the spatial coordinate $x$ as input and predicts $T$. The temperature-network architecture was selected through hyperparameter tuning using B0, which uses the prescribed ground-truth conductivity. The selected network contains four hidden layers with 48 neurons per layer and achieves a temperature RMSE of 0.0157~K, indicating sufficient capacity for the present 1D problem. The same temperature-network architecture is used across all models to ensure a consistent comparison. B5 additionally uses a smaller property subnetwork with two hidden layers and 32 neurons per layer to reconstruct the scalar conductivity function $k(T)$. The temperature network uses an initial learning rate of $10^{-3}$, while the B5 property subnetwork uses $5\times10^{-4}$. Adam is used for optimization with a maximum of 20,000 epochs. The learning rate is reduced by a factor of 0.5 when the loss evaluated on a fixed set of monitoring points stops improving. Equation~\eqref{eq:1d_governing} is enforced at 512 interior collocation points through the PDE residual loss. The prescribed end temperatures are imposed exactly through a hard boundary mapping, so no separate boundary-condition loss is required.

Three levels of conductivity-data sparsity are considered: high ($S_{\mathrm{h}}$), moderate ($S_{\mathrm{m}}$), and low ($S_{\mathrm{l}}$), corresponding to 3, 4, and 5 conductivity measurements, respectively. Four independent sampling realizations are evaluated at each sparsity level. Figure~\ref{fig:sparsity_models} shows the uniformly spaced realization for each case and the corresponding conductivity reconstructions. The dashed curve represents the hidden $k^\star(T)$, while the markers indicate the conductivity measurements available to B1--B5.

\begin{figure*}[!t]
\centering
\includegraphics[width=0.96\textwidth]{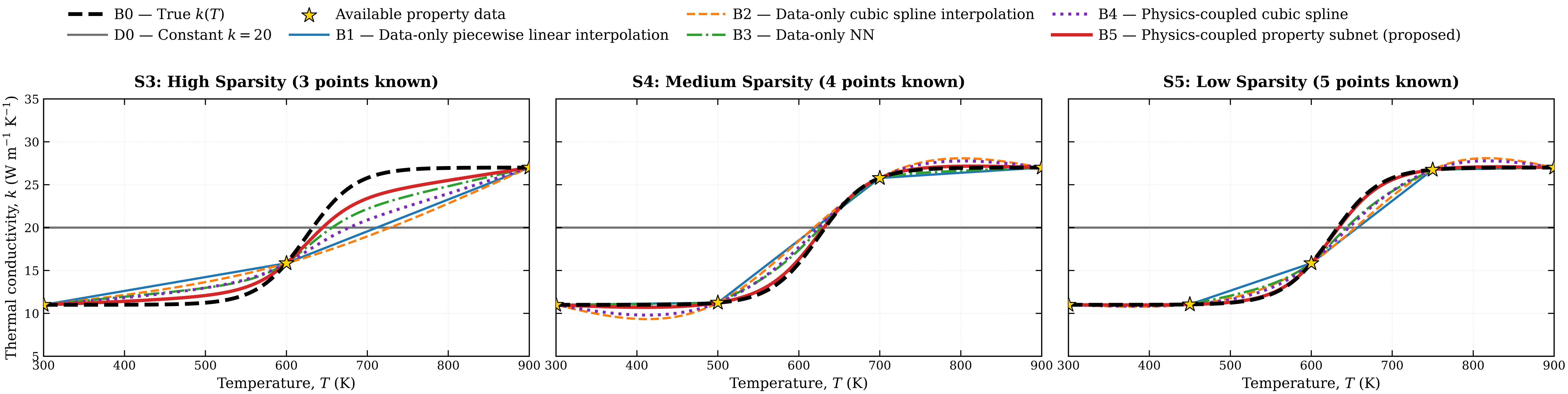}
\caption{Illustration of sparse-data realizations and conductivity reconstructions for $S_{\mathrm{h}}$, $S_{\mathrm{m}}$, and $S_{\mathrm{l}}$.}

\label{fig:sparsity_models}
\end{figure*}

For each reconstructed conductivity, an independent FDM solution is also obtained using the same $k(T)$. This solution is denoted B-FDM, while the corresponding PINN temperature solution is denoted B-PINN. The B-PINN--B-FDM discrepancy therefore measures the field-solver contribution for a fixed conductivity representation. In contrast, the error relative to the reference FDM contains the effect of conductivity reconstruction. This evaluation directly supports the error separation established in Theorem~\ref{thm:mcopinn_error}. The reported metrics are the conductivity root-mean-square error (RMSE) relative to $k^\star(T)$, the B-PINN temperature RMSE relative to the reference FDM, the direct B-PINN--B-FDM temperature RMSE, and the relative heat-flux error
\[
\frac{|q_{\mathrm{pred}}-q_{\mathrm{ref}}|}{|q_{\mathrm{ref}}|},
\]
where $q_{\mathrm{pred}}=-k(T)dT/dx$ is the heat flux predicted by the corresponding model and $q_{\mathrm{ref}}$ is the heat flux obtained from the reference FDM using the ground-truth conductivity $k^\star(T)$. Both values are evaluated as spatial averages over the domain.

B0 gives a temperature RMSE of 0.0157~K and a heat-flux error of $1.745\times10^{-5}$\%. This confirms that the state PINN accurately solves the temperature field when the complete conductivity law is prescribed. In this case, $\varepsilon_P=0$ in Theorem~\ref{thm:mcopinn_error}, leaving only the field-solver and sampling contributions.

D0 provides the complementary validation. It gives a temperature RMSE of 71.79~K and a heat-flux error of 9.09\% relative to the reference solution. However, the D0 PINN differs from its corresponding constant-property FDM solution by only $1.13\times10^{-4}$~K. Therefore, the large prediction error arises primarily from the constant-property representation rather than from the PINN solver. This remains consistent with the model discrepancy established in Theorem~\ref{thm:constant_separation} and the prediction error bound in Corollary~\ref{cor:constant_discrepancy}.

Table~\ref{tab:1d_results} summarizes the results for B1--B5 across the three sparsity levels, reported as the mean and standard deviation over four independent sampling realizations. The conductivity, temperature, and heat-flux errors are evaluated against the ground-truth conductivity and its reference FDM solution. The B-PINN vs. B-FDM RMSE instead isolates the field-solver error by comparing the PINN prediction with an FDM solution obtained using the same reconstructed conductivity. At $S_{\mathrm{h}}$, compared with B4, B5 reduces the conductivity RMSE from 2.14 to 0.96~W/(m$\cdot$K) and the temperature RMSE relative to the reference FDM from 10.53 to 4.16~K. This simultaneous reduction in property and temperature errors is consistent with Equation~\eqref{eq:property_bound} of Theorem~\ref{thm:solution_map}, which bounds the resulting state error in terms of the property reconstruction error. B5 also gives a heat-flux error of only 0.0044\%, compared with 2.32\% for B4 and approximately 6--8\% for the data-only models. Figure~\ref{fig:conduction_summary} provides a visual summary of the three primary reconstruction and prediction metrics across $S_{\mathrm{h}}$--$S_{\mathrm{l}}$.

\begin{table*}[!t]
\caption{Performance Evaluation of PINN Models in the 1D Nonlinear-Conduction Study Using Ground-Truth References and Model-Specific FDM Comparisons (Mean $\pm$ Standard Deviation)}
\label{tab:1d_results}
\centering
\footnotesize
\setlength{\tabcolsep}{3.5pt}
\begin{tabular}{clccccc}
\toprule
Sparsity & Metric & B1 & B2 & B3 & B4 & \textbf{B5 (proposed)}\\
\midrule

$S_{\mathrm{h}}$
& $k$-RMSE [W/(m$\cdot$K)]
& 2.96 $\pm$ 0.69
& 4.10 $\pm$ 1.68
& 2.79 $\pm$ 1.00
& 2.14 $\pm$ 0.58
& \textbf{0.96 $\pm$ 0.66}
\\

& $T$-RMSE vs. Reference FDM [K]
& 11.38 $\pm$ 2.94
& 17.09 $\pm$ 5.46
& 11.47 $\pm$ 3.89
& 10.53 $\pm$ 4.30
& \textbf{4.16 $\pm$ 2.67}
\\

& B-PINN vs. B-FDM $T$-RMSE [K]
& 0.23 $\pm$ 0.32
& 0.02 $\pm$ 0.02
& 0.02 $\pm$ 0.02
& 0.03 $\pm$ 0.02
& 0.04 $\pm$ 0.04
\\

& Heat-flux error [\%]
& 5.98 $\pm$ 6.58
& 8.27 $\pm$ 8.57
& 6.92 $\pm$ 4.31
& 2.32 $\pm$ 2.76
& \textbf{0.0044 $\pm$ 0.0045}
\\

\midrule

$S_{\mathrm{m}}$
& $k$-RMSE [W/(m$\cdot$K)]
& 2.45 $\pm$ 0.65
& 2.00 $\pm$ 1.49
& 2.21 $\pm$ 1.63
& 1.30 $\pm$ 0.72
& \textbf{0.42 $\pm$ 0.46}
\\

& $T$-RMSE vs. Reference FDM [K]
& 15.72 $\pm$ 12.14
& 8.34 $\pm$ 5.65
& 9.11 $\pm$ 6.13
& 5.37 $\pm$ 3.17
& \textbf{1.63 $\pm$ 1.90}
\\

& B-PINN vs. B-FDM $T$-RMSE [K]
& 0.27 $\pm$ 0.06
& 0.23 $\pm$ 0.41
& 0.04 $\pm$ 0.03
& 0.03 $\pm$ 0.02
& 0.03 $\pm$ 0.01
\\

& Heat-flux error [\%]
& 6.10 $\pm$ 2.87
& 5.46 $\pm$ 4.76
& 6.44 $\pm$ 5.51
& 1.86 $\pm$ 2.38
& \textbf{0.0022 $\pm$ 0.0016}
\\

\midrule

$S_{\mathrm{l}}$
& $k$-RMSE [W/(m$\cdot$K)]
& 3.70 $\pm$ 3.65
& 4.14 $\pm$ 5.26
& 2.31 $\pm$ 1.87
& 1.33 $\pm$ 0.73
& \textbf{0.24 $\pm$ 0.10}
\\

& $T$-RMSE vs. Reference FDM [K]
& 17.61 $\pm$ 15.74
& 17.53 $\pm$ 21.12
& 9.72 $\pm$ 8.25
& 7.22 $\pm$ 3.56
& \textbf{1.22 $\pm$ 0.94}
\\

& B-PINN vs. B-FDM $T$-RMSE [K]
& 0.27 $\pm$ 0.20
& 0.03 $\pm$ 0.03
& 0.07 $\pm$ 0.06
& 0.05 $\pm$ 0.06
& 0.06 $\pm$ 0.05
\\

& Heat-flux error [\%]
& 9.77 $\pm$ 10.58
& 11.68 $\pm$ 15.04
& 7.13 $\pm$ 6.55
& 2.16 $\pm$ 2.99
& \textbf{0.0044 $\pm$ 0.0038}
\\

\bottomrule
\end{tabular}
\end{table*}

\begin{figure*}[!t]
\centering
\includegraphics[width=\textwidth]{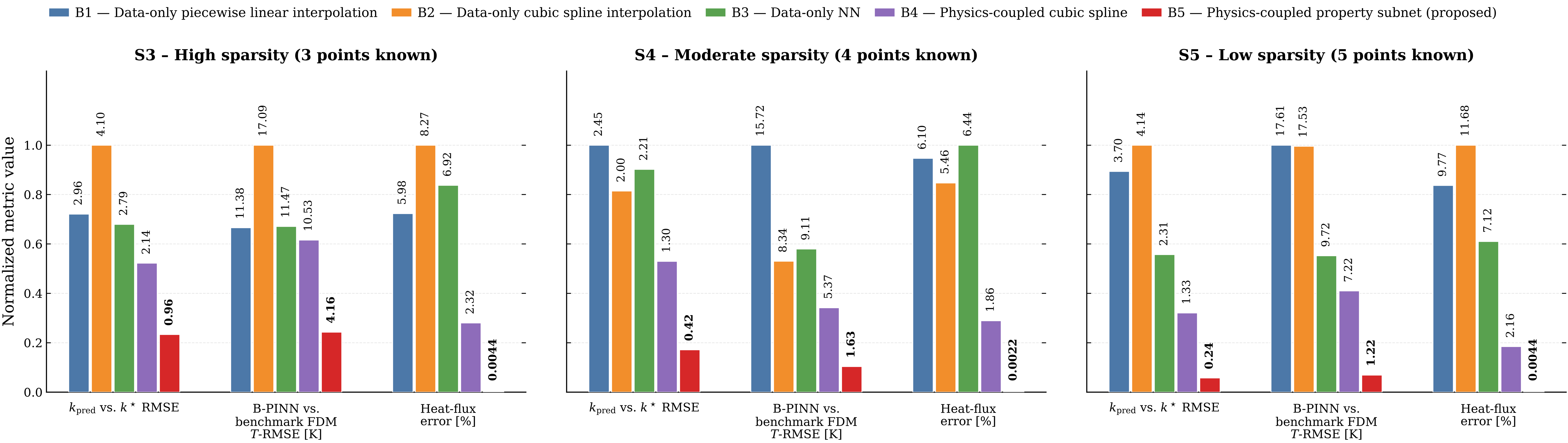}
\caption{Summary of conductivity RMSE, temperature RMSE relative to the reference FDM, and heat-flux error across the three sparsity levels.}
\label{fig:conduction_summary}
\end{figure*}

B5 improves further as additional conductivity information becomes available. At $S_{\mathrm{m}}$ and $S_{\mathrm{l}}$, the conductivity RMSE decreases to 0.42 and 0.24~W/(m$\cdot$K), while the corresponding temperature RMSE decreases to 1.63 and 1.22~K, respectively. The direct B-PINN--B-FDM discrepancies remain small, with values of 0.04, 0.03, and 0.06~K for $S_{\mathrm{h}}$--$S_{\mathrm{l}}$. This shows that, for a fixed reconstructed conductivity, the PINN closely reproduces the corresponding FDM solution and that the remaining temperature error is primarily associated with property reconstruction. This behavior is consistent with the error decomposition in Equation~\eqref{eq:mcopinn_bound} of Theorem~\ref{thm:mcopinn_error}.

B5 also maintains a heat-flux error below 0.005\% at all three sparsity levels, whereas the other reconstruction methods show noticeably larger errors. This provides an additional physics-based assessment because the heat flux depends on both the reconstructed $k(T)$ and the resulting temperature gradient. The B0 result further shows that a complete and trusted conductivity law should be prescribed directly, whereas the proposed property subnetwork is intended for cases with sparse or incomplete property information.

\subsection{MCoPINN for AFSD}
\label{subsec:mcopinn_afsd}

\subsubsection{Process Conditions and FVM Benchmark}
\label{subsubsec:afsd_conditions}

The AFSD simulations consider AA6061 over a $100\times40\times5$~mm$^3$ domain. The tool rotates at 300~rpm, and four traverse speeds from 2.12 to 8.48~mm/s are considered. The process conditions, parameters of boundary conditions, constitutive constants, and thermophysical properties are summarized in Table~\ref{tab:afsd_parameters}. The values used in the constant-property models correspond approximately to room-temperature properties and are consistent with those adopted in previous studies \cite{GOTAWALA2023114,11305145}.

\begin{table}[!t]
\caption{AFSD Process Parameters and Constants}
\label{tab:afsd_parameters}
\centering
\small
\setlength{\tabcolsep}{3pt}
\renewcommand{\arraystretch}{1.08}

\begin{tabularx}{\columnwidth}{
>{\centering\arraybackslash}p{0.15\columnwidth}
>{\raggedright\arraybackslash}X
>{\centering\arraybackslash}p{0.15\columnwidth}
>{\raggedright\arraybackslash}X}
\toprule
\textbf{Parameter} & \textbf{Value}
& \textbf{Parameter} & \textbf{Value} \\
\midrule

Domain
& $100\times40\times5$~mm$^3$
&
$(x_c,y_c)$
& $(50.125,19.875)$~mm
\\

$R_0$
& 7.16~mm
&
$R_s$
& 19.05~mm
\\

$\omega$
& 300~rpm
&
$V$
& 2.12, 4.24, 6.36, 8.48~mm/s
\\

$F$
& $0.75V$
&
$T_a$
& 298~K
\\

$\mu_f$
& 0.30
&
$\eta$
& 0.933
\\

$\lambda$
& 0.50
&
$h$
& 20~W/(m$^2\!\cdot$K)
\\

$C$
& 0.15~W/(m$^2\!\cdot$K$^3$)
&
$\beta_s$
& 0.09090~mm$^{-1}$
\\

$\chi$
& 0.90
&
$E_a$
& 205000~J/mol
\\

$R_g$
& 8.314~J/(mol$\cdot$K)
&
$\alpha_\sigma$
& 0.052~MPa$^{-1}$
\\

$n$
& 3.49
&
$\ln A$
& 27.78
\\

$\rho_C$
& 2700~kg/m$^3$
&
$k_C$
& 167~W/(m$\cdot$K)
\\

$c_{p,C}$
& 897~J/(kg$\cdot$K)
&
&
\\

\bottomrule
\end{tabularx}
\end{table}

Temperature-dependent thermophysical-property measurements for AA6061 are taken from \cite{met7070256} and are summarized in Table~\ref{tab:props_data}.

\begin{table}[!t]
\centering
\caption{Temperature-Dependent Properties of AA6061 \cite{met7070256}}
\label{tab:props_data}
\small
\renewcommand{\arraystretch}{1.10}
\setlength{\tabcolsep}{5pt}
\begin{tabular}{cccc}
\toprule
$T$ [K]
& $\rho$ [kg/m$^3$]
& $k$ [W/(m$\cdot$K)]
& $c_p$ [J/(kg$\cdot$K)]
\\
\midrule
273.0 & 2700 & 162 & 917  \\
366.3 & 2685 & 177 & 978  \\
477.4 & 2657 & 192 & 1028 \\
588.6 & 2630 & 207 & 1078 \\
700.7 & 2602 & 223 & 1133 \\
844.1 & 2574 & 253 & 1230 \\
\bottomrule
\end{tabular}
\end{table}

Linear fits to the data in Table~\ref{tab:props_data} are obtained as
\begin{equation}
\label{eq:lin_props}
\begin{aligned}
k_{\mathrm{ref}}(T)
&=119.0000+0.15384T,\\
c_{p,\mathrm{ref}}(T)
&=776.13848+0.52527T,\\
\rho_{\mathrm{ref}}(T)
&=2764.9414-0.22819T.
\end{aligned}
\end{equation}
Here, $k_{\mathrm{ref}}$, $c_{p,\mathrm{ref}}$, and $\rho_{\mathrm{ref}}$ are expressed in W/(m$\cdot$K), J/(kg$\cdot$K), and kg/m$^3$, respectively.

Two FVM formulations are then considered. The constant-property FVM uses $\rho_C$, $k_C$, and $c_{p,C}$ from Table~\ref{tab:afsd_parameters}. The linear-law FVM uses the temperature-dependent relations in Equation~\eqref{eq:lin_props} and serves as the primary numerical benchmark for the thermal and material-flow fields. The constant-property FVM is retained as a baseline for separating the effect of the thermophysical-property representation from numerical-solver effects. Both FVM formulations are executed on an Intel Core i9-14900F workstation.

\subsubsection{PINN Model Implementations}
\label{subsubsec:afsd_pinn_implementation}

Two PINN formulations are considered: the constant-property CoPINN and the proposed MCoPINN. The constant-property CoPINN uses the fixed values in Table~\ref{tab:afsd_parameters}, while MCoPINN reconstructs $k(T)$, $\rho(T)$, and $c_p(T)$ through the property subnetwork described in Section~\ref{subsec:mcopinn_formulation}.

A hyperparameter study was performed using the 300~rpm and 2.12~mm/s condition, after which the selected configuration was fixed for all traverse speeds. The flow and temperature subnetworks each contain four hidden layers with 48 neurons per layer, while the MCoPINN property subnetwork contains three hidden layers with 16 neurons per layer. All networks use tanh activation functions. Training uses 15,000 interior collocation points. The boundary conditions are enforced using 4,000 points on the top surface and 5,000 points distributed over the remaining surfaces. 

Both PINN models are trained using Adam with an initial learning rate of $2\times10^{-5}$ and a maximum budget of $10^5$ epochs. Training is evaluated on a fixed monitoring set every 1000 epochs. Training terminates after six consecutive monitoring evaluations without improvement, while the learning rate is reduced by a factor of 0.5 after three such evaluations. The same state-network configuration and training controls are retained across both PINN formulations.

Table~\ref{tab:afsd_computational_cost} summarizes the computational times. The FVM formulations require approximately 48--52~hours per operating condition, whereas the PINN-based formulations require approximately 8--8.5~hours of training using two NVIDIA GeForce RTX~2080~Ti GPUs. This corresponds to a reduction of approximately 83\% in computational time for each operating condition.

\begin{table}[!t]
\centering
\caption{Computational Time per AFSD Operating Condition}
\label{tab:afsd_computational_cost}
\footnotesize
\renewcommand{\arraystretch}{1.10}
\begin{tabular}{lc}
\toprule
\textbf{Model} & \textbf{Computational time} \\
\midrule
Constant-property FVM    & $\sim48$~hours \\
Linear-law FVM           & $\sim52$~hours \\
Constant-property CoPINN & $\sim8$~hours \\
MCoPINN (proposed)       & $\sim8.5$~hours \\
\bottomrule
\end{tabular}
\end{table}

\subsubsection{Field Verification Against FVM Benchmark and Property Reconstruction}
\label{subsubsec:field_comparison}
Field errors are evaluated within the deposition zone, $-20 \leq x-x_c \leq 20$~mm, over the full $y$ and $z$ extents. This region contains the strongest thermal and deformation gradients and prevents the weakly varying regions outside the deposition zone from dominating the error metrics.

Before comparison with the main benchmark, the constant-property CoPINN is compared with the constant-property FVM as a solver-consistency check. At 2.12~mm/s, the normalized RMSEs between the two models are 3.77\% for temperature and 4.18\% for resultant velocity, and both remain below 5\% across all four traverse speeds. Since both models use the same constant thermophysical properties and governing formulation, their agreement indicates that the PINN field-solver error is small. This comparison is therefore used only to verify the CoPINN solution accuracy and not as the primary physical benchmark.

The linear-law FVM is used as the primary numerical benchmark because it directly prescribes the temperature-dependent relations in Equation~\eqref{eq:lin_props}. Figure~\ref{fig:afsd_temperature_fields} compares the temperature fields predicted by the four formulations at 2.12~mm/s. The constant-property FVM and CoPINN predict a visibly hotter deposition zone than the linear-law FVM and MCoPINN. The constant-property FVM reaches approximately 857~K, compared with approximately 822~K for the linear-law FVM.

\begin{figure*}[!t]
\centering
\includegraphics[width=\textwidth]
{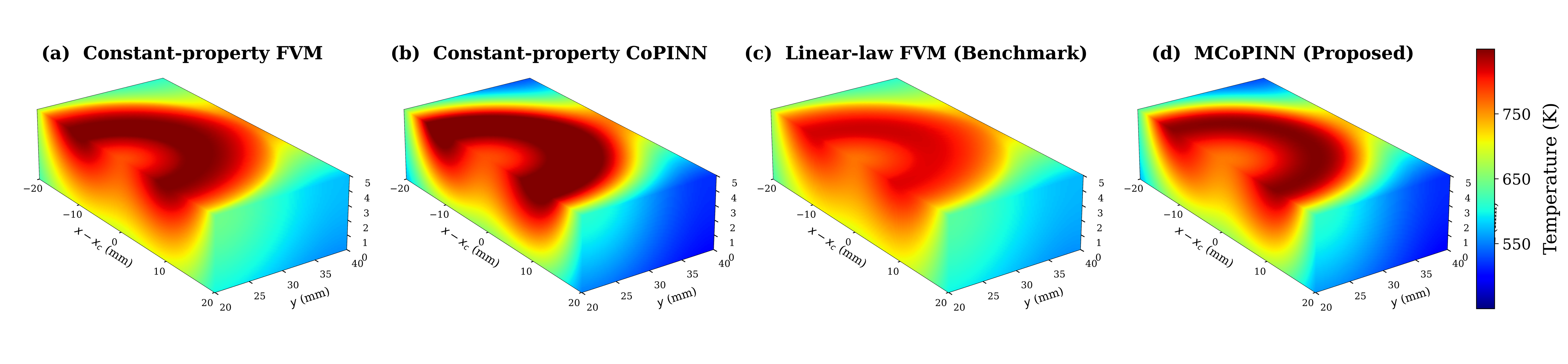}
\caption{Comparison of temperature fields at 2.12~mm/s: (a) constant-property FVM, (b) constant-property CoPINN, (c) linear-law FVM, and (d) MCoPINN (proposed).}
\label{fig:afsd_temperature_fields}
\end{figure*}

The higher temperatures predicted by the constant-property models are primarily associated with the assumption of constant thermophysical properties. This is supported by the constant-property FVM, which exhibits the same upward temperature shift. Over the temperature range reached during AFSD, thermal conductivity and specific heat increase with temperature, while density decreases. The resulting variation in thermal diffusivity,
$\alpha(T)=k(T)/[\rho(T)c_p(T)]$, enhances heat transport away from the deposition zone in the variable-property formulations and reduces the local temperature rise.

Figure~\ref{fig:afsd_velocity_fields} compares the corresponding resultant material-flow velocity fields. The differences among the four formulations are smaller than those observed for temperature because the flow near the tool is strongly constrained by the prescribed rotational, translational, and feed velocities. In addition, $k(T)$ and $c_p(T)$ enter the energy equation directly, whereas their effect on the momentum solution occurs indirectly through the temperature-dependent Sellars--Tegart constitutive response and effective viscosity. Consequently, the material-flow field is less sensitive to the thermophysical-property representation than the temperature field.

\begin{figure*}[!t]
\centering
\includegraphics[width=\textwidth]
{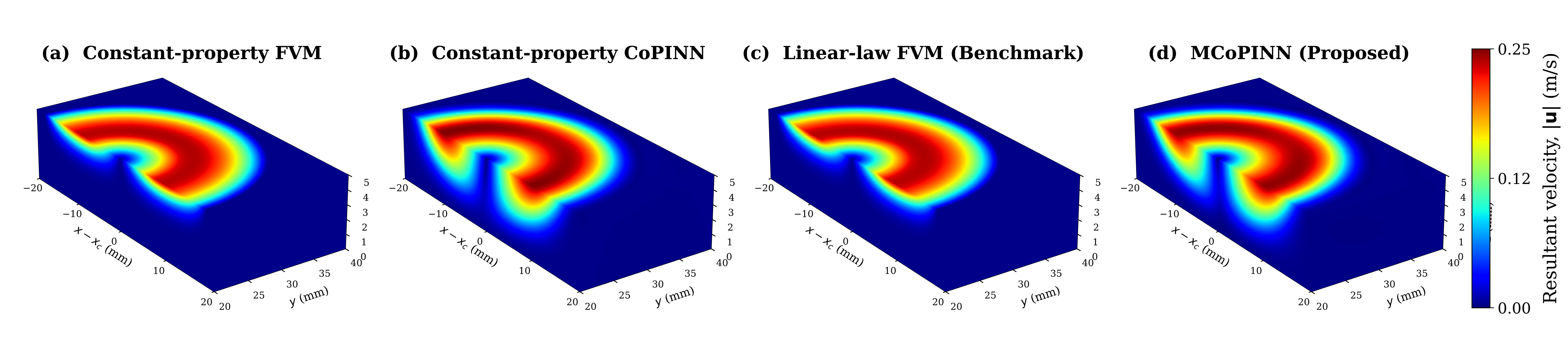}
\caption{Comparison of resultant material-flow velocity fields at 2.12~mm/s: (a) constant-property FVM, (b) constant-property CoPINN, (c) linear-law FVM, and (d) MCoPINN (proposed).}
\label{fig:afsd_velocity_fields}
\end{figure*}

\begin{table}[!t]
\caption{Temperature and Resultant-Velocity RMSEs Relative to the Linear-Law FVM Benchmark}
\label{tab:afsd_rmse}
\centering
\footnotesize
\setlength{\tabcolsep}{2pt}
\renewcommand{\arraystretch}{1.15}

\begin{tabularx}{\columnwidth}{@{}
>{\centering\arraybackslash}p{0.14\columnwidth}
>{\raggedright\arraybackslash}X
>{\centering\arraybackslash}p{0.17\columnwidth}
>{\centering\arraybackslash}p{0.19\columnwidth}
@{}}
\toprule

\parbox[c][2.8\baselineskip][c]{\linewidth}{
\centering\textbf{Speed}\\\textbf{[mm/s]}}
&
\parbox[c][2.8\baselineskip][c]{\linewidth}{
\raggedright\textbf{Model}}
&
\parbox[c][2.8\baselineskip][c]{\linewidth}{
\centering\textbf{$T$ RMSE}\\\textbf{[K]}}
&
\parbox[c][2.8\baselineskip][c]{\linewidth}{
\centering\textbf{$|\mathbf{u}|$ RMSE}\\\textbf{[m/s]}}
\\

\midrule

\multirow{2}{*}{2.12}
& Constant-property CoPINN
& 44.88
& 0.0099 \\
& MCoPINN (proposed)
& \textbf{31.97}
& \textbf{0.0091} \\

\addlinespace[4pt]

\multirow{2}{*}{4.24}
& Constant-property CoPINN
& 43.92
& 0.0084 \\
& MCoPINN (proposed)
& \textbf{24.76}
& \textbf{0.0076} \\

\addlinespace[4pt]

\multirow{2}{*}{6.36}
& Constant-property CoPINN
& 49.00
& 0.0078 \\
& MCoPINN (proposed)
& \textbf{24.91}
& \textbf{0.0064} \\

\addlinespace[4pt]

\multirow{2}{*}{8.48}
& Constant-property CoPINN
& 62.73
& 0.0071 \\
& MCoPINN (proposed)
& \textbf{39.48}
& \textbf{0.0057} \\

\bottomrule
\end{tabularx}
\end{table}

Table~\ref{tab:afsd_rmse} quantifies the field errors relative to the linear-law FVM benchmark. Averaged across the four traverse speeds, the temperature RMSE decreases from 50.13~K for the constant-property CoPINN to 30.28~K for MCoPINN. The velocity RMSE also decreases when temperature-dependent properties are included, although the improvement is smaller than for temperature. These results demonstrate that incorporating temperature-dependent thermophysical properties improves the accuracy of the predicted thermal and material-flow fields.

This interpretation is supported by the reconstructed thermophysical properties shown in Figure~\ref{fig:afsd_property_reconstruction} for 2.12~mm/s. The reconstruction RMSEs are 1.51~W/(m$\cdot$K) for $k$, 3.01~kg/m$^3$ for $\rho$, and 12.50~J/(kg$\cdot$K) for $c_p$. The corresponding normalized RMSEs are 0.61\%, 0.11\%, and 1.02\%, with $R^2$ values of 0.996, 0.994, and 0.979, respectively. Property reconstruction errors remain low across the remaining traverse speeds, indicating consistent recovery of the thermophysical-property relations across the investigated process conditions. These results are consistent with Equation~\eqref{eq:property_bound} of Theorem~\ref{thm:solution_map}, which relates property reconstruction error to the resulting state error.

For the present AA6061 case, the fitted property relations provide a reliable numerical reference through the FVM benchmark. MCoPINN reproduces the corresponding full-field behavior while reconstructing the temperature dependence from discrete property measurements. The controlled 1D study provides a complementary sparse-data setting in which the advantages of physics-coupled property reconstruction are examined against alternative reconstruction methods.

\begin{figure*}[!t]
\centering
\includegraphics[width=\textwidth]
{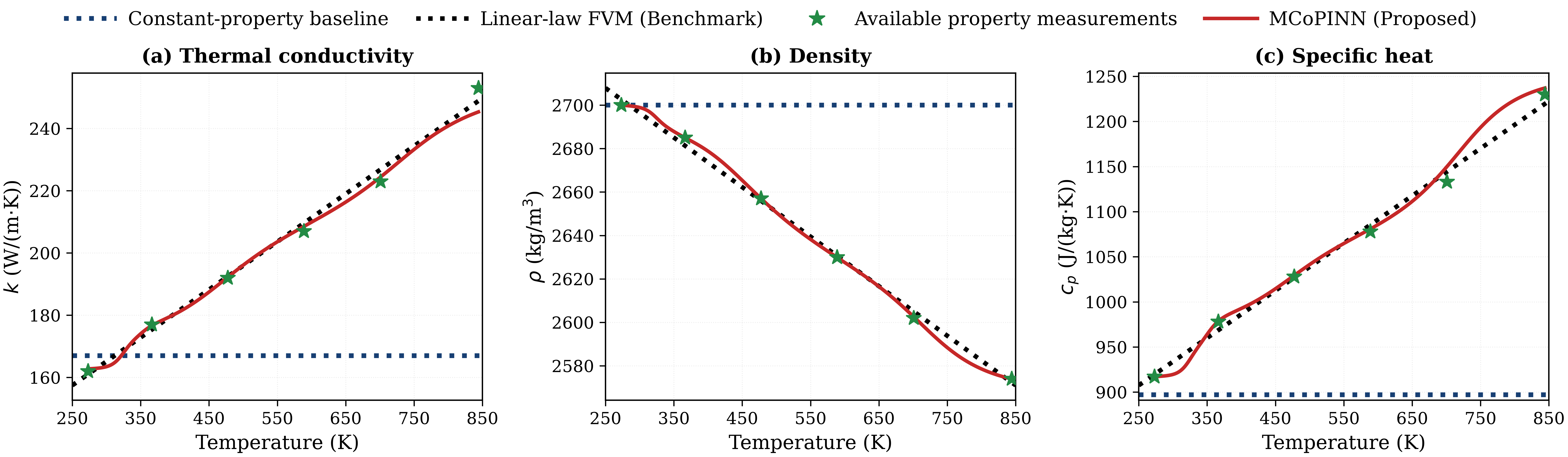}
\caption{MCoPINN reconstruction of temperature-dependent thermal conductivity, density, and specific heat capacity at 2.12~mm/s.}
\label{fig:afsd_property_reconstruction}
\end{figure*}

\subsubsection{Experimental Validation}
\label{subsubsec:thermocouple_validation}

The experimental thermocouple arrangement is shown in Figure~\ref{fig:exp_setup}. The thermocouple is located at $r_{\mathrm{TC}}=16.4$~mm and $L_{\mathrm{TC}}=0.75$~mm above the tool--material interface. Interface temperature and heat flux are sampled at 36 equally spaced circumferential locations. Experimental thermocouple measurements are reserved exclusively for validation.

\begin{figure}[!t]
\centering
\includegraphics[width=0.5\columnwidth]{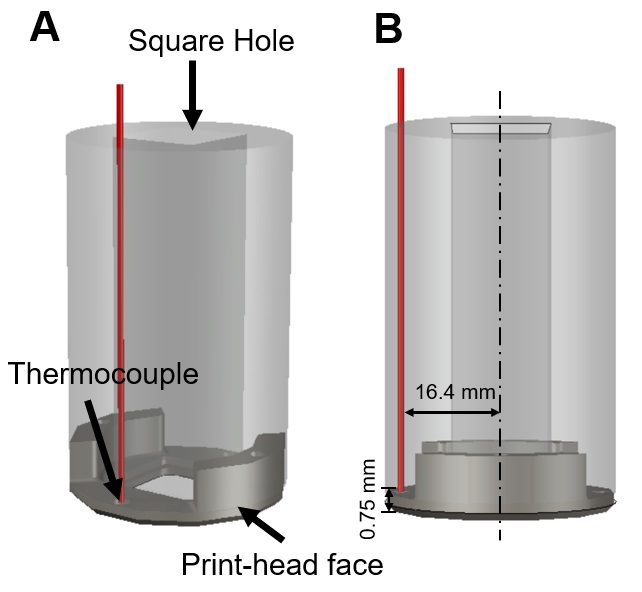}
\caption{Experimental thermocouple arrangement used for AFSD validation.}
\label{fig:exp_setup}
\end{figure}

The predicted thermocouple temperature is obtained from
\begin{equation}
\label{eq:tc_prediction}
T_{\mathrm{TC,pred}}
=
\overline{T}_{\mathrm{int}}
-
\left(
\frac{1-\eta}{\eta}
\right)
\frac{L_{\mathrm{TC}}}{k_{\mathrm{head}}}
\overline{q}_{\mathrm{int}},
\end{equation}
where $\overline{T}_{\mathrm{int}}$ and $\overline{q}_{\mathrm{int}}$ are the circumferentially averaged interface temperature and heat flux over the 36 sampling locations, respectively. The tool-head thermal conductivity is $k_{\mathrm{head}}=17.6$~W/(m$\cdot$K), and $\eta$ is the heat-partition parameter listed in Table~\ref{tab:afsd_parameters}. The same procedure is applied to all four numerical formulations.

Figure~\ref{fig:thermocouple_compare} compares the predicted and measured thermocouple temperatures across four traverse speeds. MCoPINN provides the closest overall agreement with the experimental measurements, while the constant-property models generally overpredict temperature and the linear-law FVM underpredicts it. The temperature trend is quantified using the slope of temperature against traverse speed. The measured slope is $-5.86$~K/(mm/s), while MCoPINN predicts $-5.97$~K/(mm/s), a difference of less than 2\%. The linear-law FVM also captures the decreasing trend, with a slope of $-8.66$~K/(mm/s), despite larger errors in absolute temperature predictions. In contrast, the constant-property FVM and CoPINN predict much smaller slopes of $-2.42$ and $-1.92$~K/(mm/s), respectively. This shows that the constant-property models fail to accurately capture the sensitivity of temperature to changes in traverse speed.

\begin{figure}[!t]
\centering
\includegraphics[width=0.7\columnwidth]
{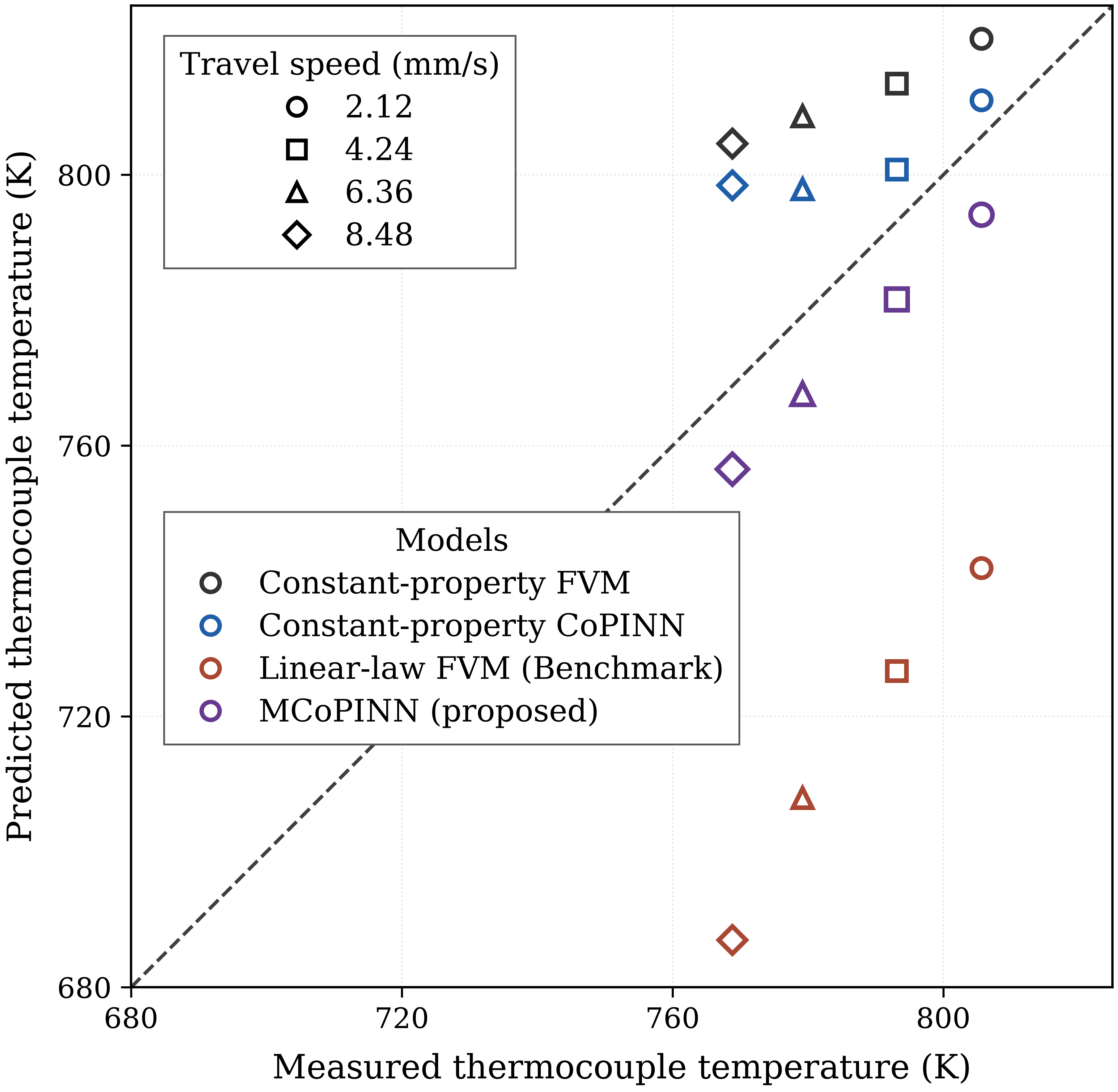}
\caption{Predicted versus measured thermocouple temperatures for the four numerical formulations across four traverse speeds.}
\label{fig:thermocouple_compare}
\end{figure}

The decrease in temperature with increasing traverse speed is consistent with the shorter interaction time between the tool and material, which leads to reduced local heat accumulation. As temperature changes, thermal conductivity, density, and specific heat also vary, affecting heat conduction and energy transport. MCoPINN accounts for these changes through the coupling between the predicted temperature and thermophysical properties. In contrast, the constant-property models cannot capture this additional temperature dependence, which helps explain their weaker response to changes in traverse speed. Their systematic temperature overprediction is particularly undesirable because AFSD is a solid-state process. Overestimating temperature can incorrectly suggest that the material approaches or exceeds its solidus, leading to misleading conclusions about the processing window and thermo-mechanical response. The experimental comparison therefore complements the full-field verification against the FVM benchmark.

\section{Conclusions and Future Work}
\label{sec:conclusion}
This work developed an MCoPINN framework for steady AFSD with temperature-dependent thermophysical properties. The theoretical analysis separated the contributions of property reconstruction error, PINN field-solver error, and constant-property model discrepancy. The controlled 1D study showed that the proposed property subnetwork is particularly effective when material information is sparse, achieving a temperature RMSE of 4.16~K and a heat-flux error below 0.005\% at the sparsest level ($S_{\mathrm{h}}$). In the AFSD application, MCoPINN reduced the average temperature RMSE from 50.13~K for the constant-property CoPINN to 30.28~K while accurately reconstructing the temperature-dependent thermophysical properties. The PINN-based formulation required approximately 8.5~hours per operating condition compared with approximately 52~hours for the benchmark FVM, corresponding to a reduction of approximately 83\% in computational time.

Future work will extend the framework to transient and multilayer cases. The extended framework will also account for uncertainties in process conditions, material properties, and measurements to assess the robustness of the reconstructed properties and predicted fields. Operator-learning approaches may also be incorporated to improve generalization across process parameters without retraining a separate model for each operating condition. Energy-based and variational formulations will be explored as possible alternatives for reducing the computational cost even further. These developments can further improve the applicability of physics-informed constitutive learning to computationally intensive additive-manufacturing processes.

\section*{Acknowledgment}
\label{sec:acknowledgment}
This work was supported in part by the U.S. Army Research Laboratory under Award W911NF-25-2-0067 and in part by the Virginia Innovation Partnership Corporation under Award CCF23-0245-ERR. The authors acknowledge the use of ChatGPT (OpenAI) for language refinement and editorial assistance during manuscript preparation. The authors take full responsibility for the technical content, analyses, and conclusions presented in this work.

\bibliographystyle{IEEEtran}
\bibliography{References2}

\end{document}